\documentclass[journal,draftcls,onecolumn,12pt,twoside]{IEEEtranTCOM}
\normalsize

\ifCLASSINFOpdf
\else
\fi

\usepackage[utf8]{inputenc} 
\usepackage[T1]{fontenc}
\usepackage{url}
\usepackage{ifthen}
\usepackage{cite}
\usepackage[cmex10]{amsmath} 
\usepackage{subcaption}
\usepackage{graphicx}
\usepackage{amsmath }
\usepackage{amssymb}
\usepackage{amsfonts}
\usepackage{mathtools}
\usepackage{dsfont}
\mathtoolsset{showonlyrefs=true}
\usepackage{algorithm}
\usepackage{algorithmicx}
\usepackage[noend]{algpseudocode}

\usepackage{color}
\usepackage{float}
\definecolor{alizarin}{rgb}{0.82, 0.1, 0.26}

\newcommand{\crc}{{\rm crc}}
\newcommand{\aug}{{\rm aug}}
\newcommand{\thr}{{\rm thr}}

\newcommand{\ch}{{\rm ch}}

\newcommand{\rem}[1]{}

\newtheorem{proposition}{Proposition}

\newcommand{\bre}{\begin{equation}}
\newcommand{\ere}{\end{equation}}
\newcommand{\be}{{\bf {e}}}
\newcommand{\ee}\]
\newcommand{\bra}{\begin{eqnarray}}
\newcommand{\era}{\end{eqnarray}}
\newcommand{\bfg}{\begin{figure}[hbtp]}
\newcommand{\efg}{\end{figure}}
\newcommand{\bver}{\begin{verbatim}}
\newcommand{\ever}{\end{verbatim}}
\newcommand{\bit}{\begin{itemize}}
\newcommand{\eit}{\end{itemize}}
\newcommand{\ben}{\begin{enumerate}}
\newcommand{\een}{\end{enumerate}}

\newcommand{\bA}{{\bf A}}
\newcommand{\bB}{{\bf B}}

\newcommand{\bI}{{\bf I}}

\newcommand{\bG}{{\bf{G}}}

\newcommand{\bF}{{\bf F}}

\newcommand{\bc}{{\bf c}}

\newcommand{\bl}{{\boldsymbol\ell}}

\newcommand{\bR}{{\bf R}}

\newcommand{\norm}[1]{\|#1\|}

\newcommand{\baa}{\begin{eqnarray*}}
\newcommand{\eaa}{\end{eqnarray*}}

\newcommand{\bs}{{\bf s}}
\newcommand{\br}{{\bf r}}
\newcommand{\ba}{{\bf a}}

\newcommand{\bhh}{{\bf h}}
\newcommand{\bh}{{\bf h}}
\newcommand{\bH}{{\bf H}}
\newcommand{\bL}{{\bf L}}

\newcommand{\bu}{{\bf u}}

\newcommand{\bd}{{\bf d}}
\newcommand{\bD}{{\bf D}}

\newcommand{\bn}{{\bf n}}
\newcommand{\bx}{{\bf x}}
\newcommand{\by}{{\bf y}}

\newcommand{\cA}{{\cal A}}
\newcommand{\cP}{{\cal P}}

\newcommand{\cL}{{\cal L}}

\newcommand{\cK}{{\cal K}}

\newcommand{\cC}{{\mathcal{C}}}

\newcommand{\cN}{{\mathcal{N}}}

\newcommand{\defined}{\triangleq}

\def\argmax{\mathop{\rm argmax}}
\def\argmin{\mathop{\rm argmin}}

\def\defined{\: {\stackrel{\scriptscriptstyle \Delta}{=}} \: }

\newfont{\boldlarge}{msbm10 scaled 1100}

\newcommand{\comment}[1]{}

\renewcommand{\thesection}{\Roman{section}}

\newlength{\tmpbigbar}

\renewcommand{\baselinestretch}{1.6}

\begin{document}
%
\title{Efficient Belief Propagation List Ordered Statistics Decoding of Polar Codes}
%
%
%

\author{Yonatan~Urman,~\IEEEmembership{Student~Member,~IEEE,}
        Guy~Mogilevsky,
        and~David~Burshtein,~\IEEEmembership{Senior~Member,~IEEE}
\thanks{This research was supported by the Israel Science Foundation (grant no. 1868/18).}%
\thanks{Y. Urman, G. Mogilevsky and D. Burshtein are with school of Electrical Engineering, Tel-Aviv University, Tel-Aviv 6997801, Israel (email: yonatanurman@mail.tau.ac.il, guym1@mail.tau.ac.il, burstyn@eng.tau.ac.il).}
\thanks{{The material in this paper will be presented in part in the International Symposium on Information Theory (ISIT), Melbourne, July 2021.}}}

\maketitle

\begin{abstract}
New algorithms for efficient decoding of polar codes (which may be CRC-augmented), transmitted over either a binary erasure channel (BEC) or an additive white Gaussian noise channel (AWGNC), are presented.
We start by presenting a new efficient exact maximum likelihood decoding algorithm for the BEC based on inactivation decoding and analyze its computational complexity. This algorithm applies a matrix triangulation process on a sparse polar code parity check matrix, followed by solving a small size linear system over GF(2).
We then consider efficient decoding of polar codes, transmitted over the AWGNC. The algorithm applies CRC-aided belief propagation list (CBPL) decoding, followed by ordered statistics decoding (OSD) of low order.
Even when the reprocessing order of the OSD is as low as one, the new decoder is shown to significantly improve on plain CBPL.
To implement the OSD efficiently, we adapt the matrix triangulation algorithm from the BEC case.
We also indicate how the decoding algorithms can be implemented in parallel for low latency decoding.
Numerical simulations are used to evaluate the performance and computational complexity of the new algorithms.
\end{abstract}


%
\IEEEpeerreviewmaketitle

\section{Introduction} \label{sec:intro}
%
%
%
%
\IEEEPARstart{P}{olar} codes \cite{PolarCodes}, decoded by the low complexity successive cancellation (SC) decoder, provably achieve channel capacity for a wide range of channels. Nevertheless, their finite length performance is typically inferior to that of low-density parity-check (LDPC) codes or Turbo codes.
The error rate performance of a polar code with a short to moderate blocklength can be significantly improved by concatenating it with a high rate
cyclic redundancy check (CRC) code, and using a CRC-aided SC list (SCL) decoder~\cite{SCL}.
However, both the SC and SCL decoders are sequential and thus suffer from high decoding latency and limited throughput.
Improvements to SC and SCL were proposed by various authors, e.g., \cite{sarkis2014fast,hashemi2018decoder}.

An iterative belief propagation (BP) decoder over the polar code factor graph (FG) was proposed in \cite{Arkan2010PolarC, polar_vs_reed} for lower latency and higher throughput. This decoder is a soft-input, soft-output decoder, such that it can be concatenated with other soft-input decoders. In addition, it is inherently parallel and allows for efficient, high throughput implementation, optimizations and extensions \cite{eslami2010on,BP_arc,Polar_LDPC_conc,bp_early_term,crc_early_term,PCForChannelSrc,PolarBPCRCWarren,BPPermuted,BPL,BPPermutedWarren,yu2019belief,PolarBPCRCBrink}.
As was mentioned above, polar codes can benefit from an outer CRC code. In \cite{crc_early_term} it was proposed to use the CRC as an early termination criterion of a BP polar code decoder. Later, the authors of \cite{PolarBPCRCWarren} proposed concatenating the CRC FG to the polar code FG, and applied message-passing decoding on the combined FG. This yields the CRC-aided BP (CBP) decoder. In \cite{PCForChannelSrc, BPPermuted, BPPermutedWarren} it was noticed that it is possible to permute the order of the stages of the polar code FG without changing the code, such that BP decoding can be implemented on any one of the permutations. The authors of \cite{BPL} proposed the BP list (BPL) decoder, which uses several independent BP decoders in parallel. Each decoder applies BP decoding on a different permutation and the final decoded codeword is chosen from the output of the decoder which has the minimal $\cL_2$ distance to the received signal. In \cite{PolarBPCRCBrink}, the authors suggested the CRC-aided BPL (CBPL) algorithm that incorporates CRC information in each of the BP decoders of the BPL algorithm.
A different approach to BP decoding of polar codes \cite{SparseGraphsBPPolar} uses a sparse polar parity check matrix (PCM), unlike the standard polar PCM which is dense \cite[Lemma 1]{LP_Polar_decoding}. The sparse PCM contains hidden variables in addition to the codeword variables. Due to the sparsity, standard BP decoding can be applied on the FG corresponding to this PCM, as for LDPC codes.

Another family of decoders for polar codes are based on ordered statistics decoding (OSD) \cite{OSD} and its box and match variant \cite{valembois2004box}. This approach was used e.g., in \cite{OSDPolar16}, and in \cite{trifonov2012efficient,goldin2019performance} for polar and polar-like codes when viewed as a concatenated code.
However, the required reprocessing order is typically large. Even though there are known methods for improving the complexity-performance trade-off of OSD, e.g. \cite{valembois2004box}, the computational complexity may still be prohibitive due to a large reprocessing order, or to the Gaussian elimination required.
In the context of OSD decoding of LDPC codes it was suggested \cite{BPOSD} to combine BP and OSD decoding by using the soft decoded BP output to rank the codebits rather than using the uncoded information as in plain OSD.

Although much progress has been made in developing efficient belief propagation decoding-based algorithms for polar codes, these algorithms still have a higher error rate compared to the CRC-aided SCL decoder.
In this work, we present the following contributions.
\begin{enumerate}
	\item
	In the first part of this work, which was initially presented in \cite{polar_ML_BEC_ISIT},
	we consider a polar code (possibly concatenated to a CRC code), transmitted over the binary erasure channel (BEC). We derive a new low complexity algorithm for computing the exact maximum-likelihood (ML) codeword based on inactivation decoding \cite{pishro2004on,LDPC_ML,shokrollahi2006systems,eslami2010on,lazaro2017inactivation,cocskun2020successive} (see also \cite[Chapter 2.6]{algebraic_coding_theory}).
	In \cite{eslami2010on}, it was proposed to use Algorithm C of \cite{pishro2004on} for improved BP decoding of polar codes over the BEC.
	The differences between our decoder for the BEC and these works are as follows.
	First, rather than using the large standard polar code FG in \cite{eslami2010on}, we use the method in \cite{SparseGraphsBPPolar} for constructing (offline) a much smaller sparse PCM for polar codes, thus reducing the computational complexity.
	In addition, we show how to extend our method to CRC-augmented polar codes, which are essential for high performance at short to moderate blocklengths.
	We also analyze the computational complexity of our method. For large $N$ the computational complexity is $\mathcal{O}(N\log N$).
	\item
	We consider the case where the (possibly CRC concatenated) polar code is transmitted over the additive white Gaussian noise channel (AWGNC), and suggest an efficient implementation of a CBPL decoder followed by OSD.
	Our method adapts the approach in \cite{BPOSD} to CBPL decoding of polar codes. We apply OSD reprocessing using the soft decoding output of each of the parallel CBP decoders in CBPL. Even when the reprocessing order of the OSD is as low as one, the new decoder is shown to significantly improve on CBPL. For computationally efficient reprocessing of order two (or higher) we suggest partial reprocessing.
	We use an efficient implementation of the Gaussian elimination required by OSD by adapting the efficient ML decoder for the BEC, presented in the first part of the work. We then suggest a variant, with an even lower computational complexity, that saves part of the computational complexity of the efficient Gaussian elimination algorithm. 
\end{enumerate}

The rest of this work is organized as follows.
In Section \ref{sec:background} we provide some background. In Section \ref{sec:EffML}, we describe the proposed algorithm for ML decoding of polar codes over the BEC, analyze its complexity and present simulation results. In Section \ref{sec:OSD} we introduce the CBPL-OSD polar decoder and propose an efficient implementation. We further present simulation results that show the improved error rate performance. Finally, Section \ref{sec:conclusions} concludes the paper.

\section{Background} \label{sec:background}
\subsection{Polar Codes} \label{sec:background_polar}
We use the following notations. We denote matrices by bold-face capital letters and column vectors by bold-face lower case letters. We denote by $\cP\left(N, K\right)$ a polar code of blocklength $N$, information size $K$, and rate $K/N$.
The length $K$ information set and length $N-K$ frozen set are denoted by $\cA$ and $\bar{\cA}$ respectively. We denote by $\tilde{\bu}$, the information bits vector of length $K$, and by $\bu$ the information word, such that $\bu_{\cA} = \tilde{\bu}$ and $\bu_{\bar{\cA}} = 0$. That is, $\bu$ is the full input (column) vector, including the frozen bits, and we assume that the frozen bits are set to zero. We denote the $N\times N$ polar full generator matrix by $\bG_{N} = \bB_{N}\bF^{\otimes n}$ \cite{PolarCodes}, where $\bB_{N}$ is the bit reversal permutation matrix and
$
\bF= \begin{scriptsize}
\begin{bmatrix}
1 & 0 \\
1 & 1
\end{bmatrix}
\end{scriptsize}
$.
The codeword is then generated using $\bc^T=\bu^T\bG_{N}$.
Denote the matrix composed of the rows of $\bG_N$ with indices in $\cA$ by $\bG_N(\cA)$. Then $\bG_N(\cA)$ is the generator matrix of the code, and $\bc^T = \bu_{\cA}^T \bG_N(\cA)$.
The standard polar code PCM can be constructed from $\bG_N$ by taking the columns that correspond to the frozen indices ($\bar{\cA}$) \cite[Lemma 1]{LP_Polar_decoding}.

The error rate performance of polar codes can be significantly improved by concatenating it with a CRC code.
Consider a CRC code which appends $r$ CRC bits to binary vectors of length $m$. This code can be represented as a linear binary block code with code rate $R_{\crc} = m / (m+r)$, a systematic generator matrix $\bG_{\crc}$ of dimensions $m \times (m+r)$, and a PCM $\bH_{\crc}$ with $r$ parity check constraints.
For a polar code $\cP(N,K)$ and a CRC code with dimension $m$ and blocklength $K=m+r$, the corresponding \textit{CRC-augmented polar code} is the linear binary block code of length $N$, dimension $m$ and generator matrix $\bG_{\aug} = \bG_{\crc} \bG_N(\cA)$. Note that the overall code rate is $R = m / N$.

\subsection{CRC-aided belief propagation list decoding} \label{sec:background_CBPL}
As was discussed in Section \ref{sec:intro}, for reduced latency high throughput decoding, a BP decoder acting on the polar code FG, shown in Fig. \ref{fig:PolarFG}, can be used \cite{Arkan2010PolarC, polar_vs_reed}.
\begin{figure}
	\centering
	\includegraphics[width=0.3\linewidth]{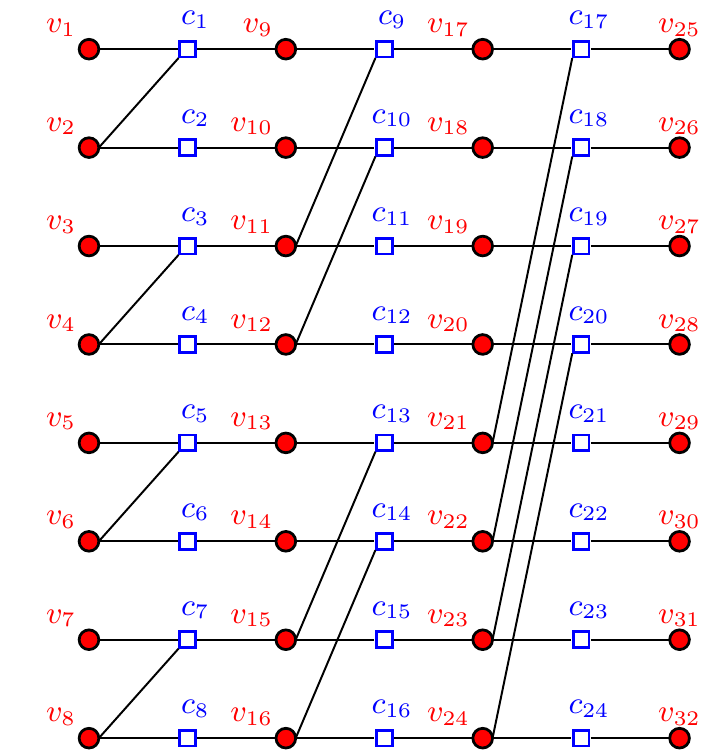}
	\caption{Polar code factor graph for n=3.}
	\label{fig:PolarFG}
\end{figure}
The polar code standard FG consists of $n=\log_2 N$ stages of parity check (PC) nodes, and $n+1$ stages of variable nodes. The variable nodes in the leftmost stage of the graph (denoted in Fig. \ref{fig:PolarFG} by red numbers $1-8$) correspond to the information word $\bu$, and the nodes on the rightmost stage of the graph correspond to the codeword $\bc$ (marked by red numbers $25-32$).
The BP decoder is iterative, where each iteration consists of a full left, followed by a full right message propagation over the FG. The information is propagated through the graph, such that a full iteration consumes $2\cdot n$ time steps.
The BP update equations are described in \cite{Arkan2010PolarC, polar_vs_reed}.
There are three types of variable nodes in the FG: Channel variable nodes (CVN), corresponding to the codeword, frozen variable nodes (FVN), and the rest are hidden variable nodes (HVN).
As an alternative to the BP equations, one can apply some variant of the min-sum algorithm, which is more convenient to apply in practice, although for the BEC, min-sum coincides with BP.

To reduce the error-rate, it was suggested to perform the message passing on the concatenation of the CRC FG, represented by the PCM $\bH_{\crc}$, and the polar FG.
This means that we add the CRC FG to the left of the polar FG shown in Fig. \ref{fig:PolarFG} such that it connects to the non-frozen variable nodes on the left side of the graph (the non-frozen variable nodes among $v_1$-$v_8$). We use the concatenated graph when applying the BP. 
As stated in \cite{PolarBPCRCWarren}, for this concatenation to be beneficial in terms of error-rate performance, the messages must first evolve through $I_{\thr}$ iterations on the polar FG alone, so that the information layer LLRs will become reliable enough. This results in the CBP decoder \cite{PolarBPCRCWarren}.
The BP iterations are performed until either a predetermined maximum number of iterations $I_{\max}$ have been reached, or until a certain stopping condition has been met.
Following \cite{bp_early_term,crc_early_term}, in our implementation we stop iterating if the following two conditions are satisfied. The first condition is
$\hat{\bc}^T=\hat{\bu}^T\cdot\bG_N$
where $\hat{\bu}$ and $\hat{\bc}$ are the hard-decisions of the polar information and codeword bits, respectively.
The second condition is that $\hat{\bu}_{\cA}$ satisfies the CRC constraints ($\bH_{\crc}\hat{\bu}_{\cA}=\mathbf{0}$).

To further reduce the error-rate, we can perform CBP on a list of $L$ layer-permuted polar FGs \cite{PCForChannelSrc,BPPermuted,BPL}, thus obtaining the CBPL decoder \cite{PolarBPCRCBrink}.
The estimated codeword is the CBP output $\hat{\bx}$ with BPSK representation closest, in Euclidean distance, to the channel output, $\by$, out of all the outputs that are valid codewords (if there are no CBP outputs $\hat{\bx}$ which are valid codewords, we choose that CBP output $\hat{\bx}$ which is closest to the channel output, $\by$).
To simplify the implementation, instead of using different polar FGs in the CBP realizations, we may permute their inputs and outputs \cite{BPPermutedWarren}.

\subsection{Sparse parity check matrix for polar codes}
We briefly review the method \cite{SparseGraphsBPPolar} for obtaining a sparse representation of the polar code PCM. It starts with the standard polar FG (Fig. \ref{fig:PolarFG}), which can be represented as a PCM of size $N\log_2 N\times N(1+\log_2 N)$, i.e., $N\log_2 N$ PC nodes, and $N(1+\log_2 N)$ variable nodes (out of which, $N-K$ variable nodes are frozen). The leftmost layer of variable nodes of the graph corresponds to columns $1$ to $N$ of the PCM, the next layer of variable nodes corresponds to columns $N+1$ to $2N$, and so on, such that the rightmost variable nodes (CVNs) correspond to the last $N$ columns of the PCM. The column index of variable node $v_i$ in Fig. \ref{fig:PolarFG} is $i$. The resulting PCM is sparse since each PC is connected to at most $3$ variable nodes, as shown in Fig. \ref{fig:PolarFG}. Hence, standard BP can be used effectively on this graph with good performance. Unfortunately, the dimensions of the resulting PCM are large (instead of $(N-K)\times N$ for the standard polar code PCM, we now have a matrix of size $N\log_2 N \times N(1+\log_2 N)$ as it contains HVNs as well), which increases the decoding complexity. To reduce the matrix size, the authors of \cite{SparseGraphsBPPolar} suggested the following pruning steps that yield a valid sparse PCM to the code, while reducing the size of the $N\log_2 N\times N(1+\log_2 N)$ original PCM significantly:
\begin{enumerate}
	\item \textbf{FVN removal}: If a variable is frozen then it is equal to zero. Thus, all columns that correspond to frozen nodes can be removed from the PCM.
	\item \textbf{Check nodes of degree 1}: The standard polar FG contains check nodes of degrees 2 and 3 only. However, after applying the pruning algorithm, we might get check nodes with degree 1. Their neighboring variable node must be 0. Thus, this check node and its single neighbor variable node can be removed from the FG (the column corresponding to that variable node is removed from the PCM).
	\item \textbf{A CVN connected to a degree 2 check node with an HVN}: The CVN must be equal to the HVN. Thus, the connecting check node can be removed, and the HVN can be replaced (merged) with the CVN.
	\item \textbf{HVN of degree 1}: This HVN does not contribute to the decoding of the other variables since its value is unknown. Hence, the check node connected to this HVN does not carry information, so that both of them can be removed from the graph.
	\item \textbf{HVN of degree 2}: This HVN can be removed and the two connected check nodes can be merged.
	\item \textbf{Degree 2 check node that is connected to two HVNs}: The two HVNs must be equal. Hence, the check node can be removed and the two HVNs can be merged.
\end{enumerate}
Iterating over the above mentioned pruning steps until convergence results in the pruned PCM (FG), which is a valid PCM for the code. As mentioned before, the last $N$ variable nodes of the resulting PCM correspond to the codeword bits (CVNs), while the rest are HVNs. Now, standard BP decoding can be used on the new FG. It is worth mentioning that although the pruning steps keep the validity of the PCM for the code, they do change the message passing scheduling due to the removal and merging of nodes. However, in the BEC case the result of BP decoding over a FG is invariant to the scheduling used, so there is no change in the decoding result when using the pruned FG instead of the original FG.
The following proposition will be useful later on.
\begin{proposition}
	The pruned PCM has full row rank.
	\label{prop:pruned_PCM_full_rank}
\end{proposition}
\begin{proof}
	As was described above, we start the pruning process with a PCM of size $N\log_2 N\times N(1+\log_2 N)$ with $N-K$ frozen variable nodes. After performing the first pruning step, FVN removal, the dimensions of the PCM are $N\log_2 N \times (N\log_2 N + K)$. When we proceed with the pruning, it can be verified, for each of the pruning steps summarized above, that for each variable node removed from the graph, exactly one check node is removed as well. Hence, at the end of the pruning process, the pruned PCM has dimensions $(N'-K) \times N'$, where $N'\ge N$ is the total number of variable nodes in the pruned graph. Now, $K$ is the dimension of the polar code. Hence the $N'-K$ rows of the pruned PCM must be linearly independent. That is, the pruned PCM has full row rank.
\end{proof}

As an example, for $\cP\left(256, 134\right)$ ($\cP\left(512, 262\right)$, respectively), the algorithm yields a pruned PCM with blocklengh $N'=355$ ($N'=773$) and the fraction of ones in the PCM is $0.7\%$ ($0.33\%$).

\section{Efficient ML Decoding over the BEC} \label{sec:EffML}
Consider the BEC, where a received symbol is either completely known or completely unknown (erased). We denote the erasure probability by $\epsilon$. The channel capacity is $C(\epsilon) = 1-\epsilon$.
Denote the input codeword by $\bc$, and the BEC output by $\by$. Since $\bc$ is a codeword it must satisfy $\bH\bc = \mathbf{0}$ where $\bH$ is a PCM of the code.
Denote by $\cK$ the set of known bits in $\bc$ (available from the channel output, $\by$), and by $\bar{\cK}$ the set of erasures.
Also denote by $\bH_{\cK}$ ($\bH_{\bar{\cK}}$, respectively) the matrix $\bH$ restricted to columns in $\cK$ ($\bar{\cK}$), so that
$
\mathbf{0} = \bH\bc = \bH_{\cK}\bc_{\cK} + \bH_{\bar{\cK}}\bc_{\bar{\cK}}
$. Hence,
\begin{equation}\label{eq:MLDecBEC}
\bH_{\bar{\cK}}\bc_{\bar{\cK}} = \bH_{\cK}\bc_{\cK} \: .
\end{equation}
As a result, it can be seen that ML decoding over the BEC is equivalent to solving the set \eqref{eq:MLDecBEC} for $\bc_{\bar{\cK}}$ \cite{modern_coding_theory}. However, the required complexity when using Gaussian elimination is $\mathcal{O}(N^3)$. We now present a much more efficient ML decoding algorithm. This algorithm uses the sparse representation of the polar PCM \cite{SparseGraphsBPPolar} that was reviewed in Section \ref{sec:background}, and is assumed to be obtained offline. Our algorithm is a modified version of the efficient ML decoder \cite{LDPC_ML}, that was proposed for efficient ML decoding of LDPC codes over the BEC.

Since we start with a sparse PCM (the pruned PCM), it may be convenient, for computational efficiency, to store this matrix by the locations of ones at each row and column (along with the number of ones). It may also be convenient to address the rows and columns using some permutation (although this may not be the preferred option for parallel implementation).
For clarity, we describe the algorithm over the PCM. However, for efficient implementation of some of the stages, some other representation may be preferable (e.g., using the FG for stage 1 described below).
The algorithm has the following stages:
\begin{enumerate}
	\item \textbf{Standard BP decoding}. Given the BEC output, apply standard BP decoding on the pruned PCM until convergence. If BP decoding was successful (all variable nodes were decoded) then we return the decoded word and exit. Otherwise, we proceed to the next decoding stage. The PCM after BP decoding is shown in Fig \ref{fig:PCMAfterBp}. In the end of the decoding we have permuted the rows and columns of the PCM such that the first columns correspond to the decoded variable nodes, and the first rows correspond to the decoded PCs (a PC node is said to be decoded if all its neighbor variable nodes have been decoded). We denote the number of decoded variable (PC, respectively) nodes at the end of this stage by $n_d$ ($n_c$). As can be seen in Fig. \ref{fig:PCMAfterBp}, the upper right part of the matrix is filled with zeros, since there are no undecoded variable nodes associated with decoded PCs.
	\begin{figure}
		\centering
			\hspace*{-2.7cm}\includegraphics[width=0.5\linewidth]{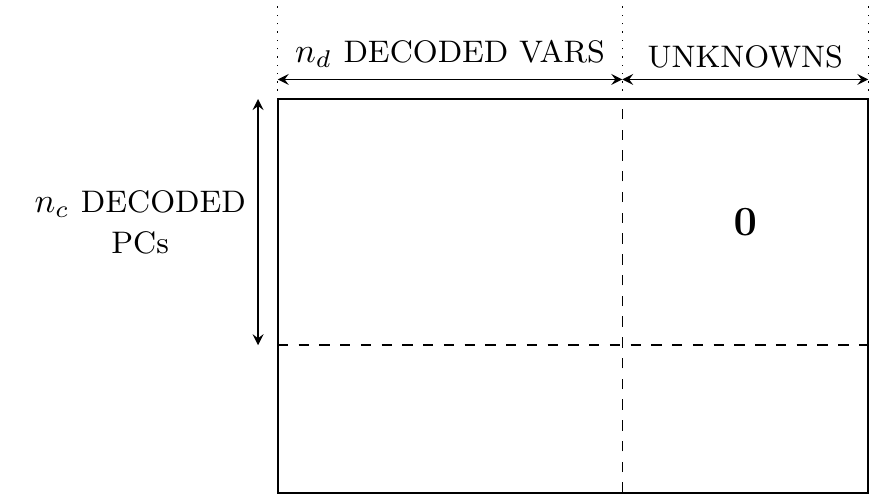}
			\caption{Reordered PCM after initial BP decoding.}
			\label{fig:PCMAfterBp}
	\end{figure}
	
	\item \textbf{Choosing reference variables and performing \emph{triangulation} \cite{EffEncLDPC}}. Consider the PCM at the output of the previous stage, shown in Fig. \ref{fig:PCMAfterBp}. Since the BP decoder has converged on this PCM, the number of unknown variable nodes in each undecoded row is larger than one (otherwise the BP could have continued decoding). The goal of the current stage is to bring the PCM to the form shown in Fig. \ref{fig:FinalPCM}, using only row and column permutations, where the $n_c \times n_d$ sub-matrix on the top left corner is the same as in Fig. \ref{fig:PCMAfterBp}, and where the sub-matrix $\bH^{(1,3)}$ is square lower triangular with ones on its diagonal.
	\begin{figure}
		\centering
		\hspace*{-3cm}\includegraphics[width=0.5\linewidth]{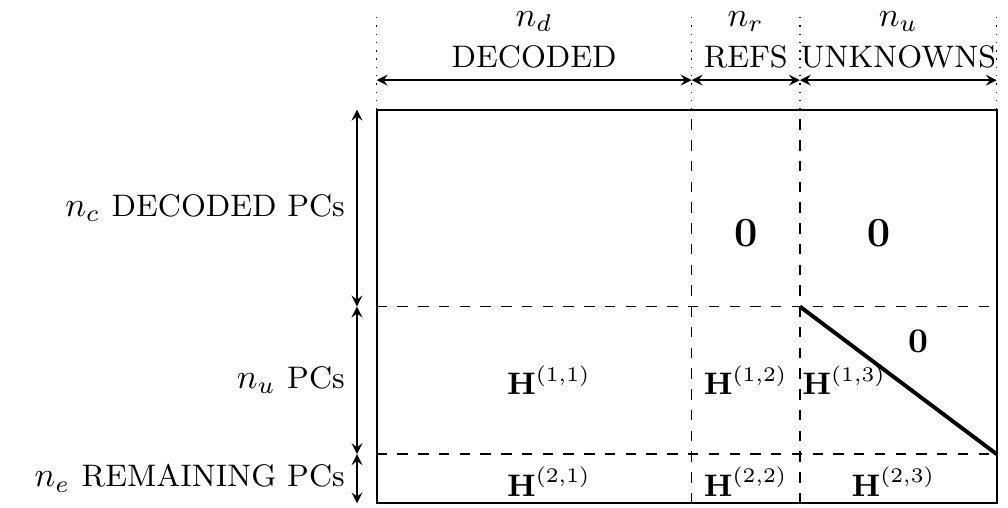}
		\caption{Final PCM. $H^{(1,3)}$ is a $n_u \times n_u$ lower triangular matrix, the bold diagonal is filled with ones.}
		\label{fig:FinalPCM}
	\end{figure}
	\par We start by considering the reordered PCM in Fig. \ref{fig:PCMAfterBp}. We mark $n'_r$ unknown variable nodes (variables that have not been decoded by the BP) as \emph{reference variables} (either by picking them at random or by using a more educated approach as discussed below) and remove them from the list of unknowns. We then permute the matrix columns so that the $n'_r$ columns corresponding to these reference variables are placed immediately after the columns corresponding to the $n_d$ decoded variables. We then perform a \emph{diagonal extension step} \cite{EffEncLDPC} on this column permuted PCM. This means that we check for undecoded rows with a single unknown variable node in the remaining columns (those with indices larger than $n_d+n'_r$), and permute the PCM to start building a diagonal just below the $n_c$ rows corresponding to decoded PCs, as shown in Fig. \ref{fig:diag_exten}.
	\begin{figure}
		\centering
		\includegraphics[width=0.7\linewidth]{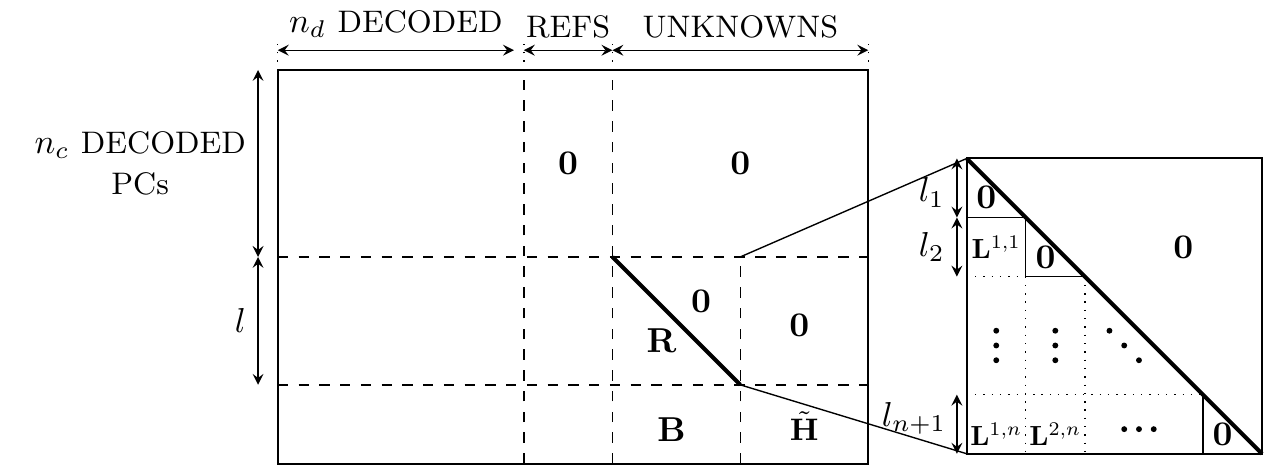}
		\caption{Diagonal Extension, the bold diagonal is filled with ones. On the right we show a zoom into the sub-matrix $\bR$, which is where the diagonalization takes place.}
		\label{fig:diag_exten}
	\end{figure}
	That is, assume we found $l_1$ such rows with a single one in the unknown variable nodes, and the locations of ones in these rows are ${(r_1, c_1), (r_2, c_2), \ldots , (r_{l_1}, c_{l_1})}$. Then, for every $i\in \{1, \ldots, l_1\}$ we permute row $r_{i}$ with row $n_c+i$, and column $c_{i}$ with column $n_d+n'_r+i$. The result is shown in Fig. \ref{fig:diag_exten} for $n=0$, $l=l_1$, $\bR=\bI$ (i.e., after the first step, there are only zeros below the diagonal) and $n'_r$ reference variables. We proceed by applying additional diagonal extension steps on the rows of the obtained matrix below the $n_c+l_1$ first ones ($\tilde{\bH}$ in Fig. \ref{fig:diag_exten}, $n=0$ and $l=l_1$) repeatedly, until this is not possible anymore (no more rows with a single unknown variable node that has not been diagonalized yet). By doing so we extend the diagonal of size $l=l_1$ that we have already constructed. Denote the total number of found rows with a single unknown variable node by $l \ge l_1$. For example, the resulting PCM for $n+1$ diagonal extension steps is also shown in Fig. \ref{fig:diag_exten} but now for $l=\sum_{i=1}^{n+1}l_i$ and some matrices $\{\bL^{i, j}\}$ for $i, j \in \{1,2, \ldots, n\}$. Essentially, after we have chosen $n'_{r}$ reference variable nodes, we have applied a BP decoding procedure that incorporates row and column permutations until convergence of the BP decoder in order to obtain the matrix shown in Fig. \ref{fig:diag_exten} with the largest possible $l$. We now repeat this process of choosing additional reference variables, permuting the matrix such that the columns corresponding to these reference variables are placed next to the previously chosen reference columns, and continuing the triangulation in order to further increase $l$, until we obtain the matrix shown in Fig. \ref{fig:FinalPCM}, where $\bH^{(1,3)}$ is square lower triangular with ones on its diagonal.
	This stage can be summarized as follows. We iterate over the following basic procedure:
	\begin{enumerate}
		\item Choose $n'_r$ additional reference variables from the unknown variables that have not been diagonalized yet.
		\item Permute the PCM columns so that the columns corresponding to these additional reference variables are placed just after the columns of the reference variables from previous applications of the procedure.
		\item Apply as many diagonal extension steps as possible.
	\end{enumerate}
	As will be discussed below, a major desirable property of the process is that typically the total number of reference variables used, $n_r$, is small.
	
	\item \textbf{Expressing unknown variables as a combination of reference variables}. When we enter this stage the PCM has the form shown in Fig. \ref{fig:FinalPCM}. It contains $n_d$ decoded variables (that have been obtained in the first BP stage), whose value is known and denoted by $\bd$, $n_r$ reference variables and $n_u$ remaining unknown variables. Denote the value of the reference variables by $\br$ and the value of the unknown variables by $\bu$. Both $\br$ and $\bu$ are unknown at this point. In this stage we express $\bu$ as an affine transformation of $\br$ over GF(2), i.e.,
	\begin{equation}
	\bu=\bA\br + \ba
	\label{eq:u_Ar_b}
	\end{equation}
	where $\bA$ is a matrix of size $n_u \times n_r$, and $\ba$ is a vector of length $n_u$. Due to the triangulation and the sparsity of the PCM, this step can be computed efficiently using back-substitution as follows. Using the matrix form in Fig. \ref{fig:FinalPCM}, we have
	\begin{equation}
	\left( \begin{matrix}
	\bs^{(1)} \\ \bs^{(2)}
	\end{matrix} \right)
	=
	\left( \begin{matrix}
	\bH^{(1,2)} & \bH^{(1,3)} \\
	\bH^{(2,2)} & \bH^{(2,3)}
	\end{matrix} \right)
	\left( \begin{matrix}
	\br \\ \bu
	\end{matrix} \right)
	\label{eq:mat_eq_s1_s2}
	\end{equation}
	where $\bs^{(1)} = \bH^{(1,1)} \bd$ and $\bs^{(2)} = \bH^{(2,1)} \bd$. Hence,
	\begin{equation}
	\bH^{(1,3)} \bu = \bH^{(1,2)} \br + \bs^{(1)}\ .
	\end{equation}
	Since $\bH^{(1,3)}$ is a lower triangular, $n_u \times n_u$ matrix, with ones on its diagonal, we thus have for $l=1,2,...,n_u$,
	\begin{equation}
	u_l = s_l^{(1)} + \sum_j H_{l,j}^{(1,2)} r_j + \sum_{j<l} H_{l,j}^{(1,3)} u_j.
	\label{eq:u_l}
	\end{equation}
	Suppose that we have already expressed $u_i$ as
	\begin{equation}
	u_i = \sum_j A_{i,j} r_j + a_i
	\label{eq:u_i}
	\end{equation}
	for $i=1,\ldots,k$, and wish to obtain a similar relation for $i=k+1$. Then, by \eqref{eq:u_l},
	\begin{equation}
	u_{k+1} = s_{k+1}^{(1)} + \sum_j H_{k+1,j}^{(1,2)} r_j + \sum_{i\in\cC_k} u_i
	\end{equation}
	where $\cC_k \defined \left\{ i \: : \: i\le k, H_{k+1,i}^{(1,3)}=1 \right\}$. Substituting \eqref{eq:u_i} for $u_i$ in the last summation and rearranging terms yields,
	\begin{equation}
	u_{k+1} = \sum_j A_{k+1,j} r_j + a_{k+1}
	\label{eq:u_k1}
	\end{equation}
	where
	\begin{equation}
	A_{k+1,j} = H_{k+1,j}^{(1,2)} + \sum_{i\in\cC_k} A_{i,j},\quad
	a_{k+1}   = s_{k+1}^{(1)} + \sum_{i\in\cC_k} a_{i}.
	\end{equation}
	This shows that the rows of $\bA$ as well as the elements of the vector $\ba=(a_1,\ldots,a_{n_u})^T$ can be constructed recursively for $k=1,2,..$ as follows. Denote the $k$'th row of $\bA$ by $\overline{\ba}_k$ and the $k$'th row of $\bH^{(1,2)}$ by $\overline{{\bhh}}^{(1,2)}_k$. Then, we initialize the recursion by,
	\begin{equation}
	\overline{\ba}_1 = \overline{\bhh}_1^{(1,2)}, \quad a_1 = s_1^{(1)} \label{eq:a_init}.
	\end{equation}
	Then, for $k=1,\ldots,n_u-1$,
	\begin{equation}
	\overline{\ba}_{k+1} = \overline{\bhh}_{k+1}^{(1,2)} + \sum_{i\in\cC_k} \overline{\ba}_{i}, \quad
	a_{k+1}              = s_{k+1}^{(1)} + \sum_{i\in\cC_k} a_{i}. \label{eq:a_recursion}
	\end{equation}
	
	\item \textbf{Finding the value of the reference and unknown variables}. By \eqref{eq:mat_eq_s1_s2} and \eqref{eq:u_Ar_b},
	it follows that $\br$ is obtained by solving the linear equation
	\begin{equation}
	\left( \bH^{(2,2)} + \bH^{(2,3)} \bA \right)\br = \bs^{(2)} + \bH^{(2,3)} \ba.
	\label{eq:linear_system_Ar_b}
	\end{equation}
	To solve the linear system \eqref{eq:linear_system_Ar_b} we use Gaussian elimination over GF(2) on the $n_e \times (n_r+1)$ augmented matrix corresponding to \eqref{eq:linear_system_Ar_b}, where $n_e$ is the number of remaining PC rows of the final PCM in Fig. \ref{fig:FinalPCM}.
	Note that this system must have a valid solution (the true transmitted codeword). Decoding will be successful if and only if the true transmitted codeword is the unique solution.
	In Section \ref{sec:EffMLComp} we will show that $n_r$ and $n_e$ are small. Hence, the complexity of this stage is also small.
	
	After we have obtained $\br$, we can also obtain $\bu$ from \eqref{eq:u_Ar_b}. Thus, we have obtained the decoded codeword.
\end{enumerate}

\subsection{Complexity} \label{sec:EffMLComp}
Recall that the pruned PCM is obtained offline. Hence, this pre-processing task has no impact on the decoding complexity. We now analyze the complexity of each stage of the efficient ML decoding algorithm described above in terms of the number of XORs (additions over GF(2)).
\begin{enumerate}
	\item
	The complexity of BP decoding over the BEC is determined by the number of edges in the Tanner graph \cite{modern_coding_theory}. We start with a PCM with $\mathcal{O}(N\log_2 N)$ edges (corresponding to the FG shown in Fig. \ref{fig:PolarFG}). Then, after applying the pruning, the total number of edges decreases. Hence, the complexity of the BP decoding on the pruned PCM is $\mathcal{O}(N\log_2 N)$.
	\item
	A diagonal extension step is equivalent to a BP iteration over the BEC with permutations instead of XORs: 
	For each undecoded row with a single unknown variable node, we apply one row and one column permutation rather than computing the XOR of all the known variable nodes in that row (in BP). Hence, the total number of permutations is of order $\mathcal{O}(N\log_2 N)$ as well.
	\item
	We first compute $\bs^{(1)} = \bH^{(1,1)} \bd$, where the complexity (number of XORs) is the number of ones in $\bH^{(1,1)}$. We then apply the recursion described by \eqref{eq:a_init} and \eqref{eq:a_recursion}. For every 1 entry in $\bH^{(1,3)}$ which is not on the diagonal, \eqref{eq:a_recursion} requires $n_r + 1$ XORs. In total, this step requires $(n_r+1)(\gamma - n_u)$ XORs, where $\gamma$ is the number of ones in $\bH^{(1,3)}$. Denote by $d_c^{(1)}$ the average number of ones in the rows of the final PCM, described in Fig. \ref{fig:FinalPCM}, corresponding to $\bH^{(1,1)}$, $\bH^{(1,2)}$ and $\bH^{(1,3)}$. Since the final PCM is sparse (it was obtained from the pruned PCM by row and column permutations only), $d_c^{(1)}$ is small. The total computational cost of this stage is $\mathcal{O}(d_c^{(1)} \cdot (n_r+1) \cdot n_u)$. We can also say that the total computational cost of this stage is $\mathcal{O}(d_c^{(1)} \cdot (n_r+1) \cdot N \log_2 N)$ (since $n_u=\mathcal{O}(N \log_2 N)$).
	\item
	We first compute $\bs^{(2)} = \bH^{(2,1)} \bd$, where the complexity is the number of ones in $\bH^{(2,1)}$. We then compute $\bH^{(2,2)} + \bH^{(2,3)} \bA$ and $\bs^{(2)} + \bH^{(2,3)} \ba$, required in \eqref{eq:linear_system_Ar_b}. The complexity is $\mathcal{O}(\rho (n_r+1))$, where $\rho$ is the number of ones in $\bH^{(2,3)}$. Let $\rho = d_c^{(2)} n_e$, where $d_c^{(2)}$ is the average number of ones in the $n_e$ remaining PC rows of the final PCM. Hence, the above complexity is $\mathcal{O}(d_c^{(2)} \cdot (n_r+1) \cdot n_e)$. Finally, we solve the linear system \eqref{eq:linear_system_Ar_b} with $n_e$ equations and $n_r$ unknowns, using Gaussian elimination. The total complexity of the Gaussian elimination is $\mathcal{O}(n_e\cdot n_r^2)$. After we have obtained $\br$, we use \eqref{eq:u_Ar_b} to obtain the unknown variables, $\bu$. The complexity of this additional step is $\mathcal{O}(n_u\cdot n_r)$.
\end{enumerate}
As can be seen, the complexity of the algorithm will be strongly influenced by the number of reference variables, $n_r$, and the number of remaining PCs, $n_e$. We argue the following.
\begin{proposition} \label{prop_expected_nr}
	When the code rate is below channel capacity, i.e., $R<C(\epsilon)=1-\epsilon$, we have,
	for any $q, p > 0$,
	\bre
	\lim_{N\rightarrow \infty} {\rm E} \{\left(N \log N\right)n_r^{q} \cdot n_e^{p}\} = 0 \: .
	\label{eq:prop_expected_nr}
	\ere
\end{proposition}
\begin{proof}
	By \cite[Lemma 6]{PCForChannelSrc}, when decoding a polar code transmitted over the BEC, BP on the standard polar FG cannot perform worse than SC. In addition, applying the BP decoder on the standard polar FG is equivalent to its application on the pruned graph (although the pruning steps have changed the message passing schedule, in the BEC case, the result of BP is invariant to the scheduling used). Now, the block error probability when using SC decoding is bounded by $P_e \leq 2^{-N^{\beta}}$, for any $\beta<1/2$, \cite{arikan2009on}. Hence, the error probability of the BP algorithm applied in the first stage of our proposed decoder is also bounded by the same term. Now, whenever BP fails to decode all bits, $n_r,n_e$ are bounded by $N \log N$. Thus, ${\rm E} \{\left(N \log N\right)n_r^q \cdot n_e^p\} \leq \left(N \log N\right)^{1 + p + q}\cdot 2^{-N^{\beta}}$. This immediately yields \eqref{eq:prop_expected_nr}.
\end{proof}
\begin{proposition} \label{prop_N_logN}
Suppose that the PCM pruning algorithm \cite{SparseGraphsBPPolar} is modified such that the maximum degree of each PC node is at most some constant $d$.
Then the average computational complexity of the decoding algorithm is $\mathcal{O}(N\log N)$.
\end{proposition}
The proof follows immediately by the computational complexity analysis above together with Proposition \ref{prop_expected_nr}.
We note that the modification in the pruning algorithm required by Proposition \ref{prop_N_logN} can be easily implemented, but is not required in practice since even without this modification, the PCM obtained is sparse. 

This $\mathcal{O}(N\log N)$ complexity should be contrasted with that of straight-forward ML decoding \eqref{eq:MLDecBEC} which is $\mathcal{O}(N^3)$ ($N \cdot\epsilon$ variables where $\epsilon$ is the erasure probability and $N-K$ equations).

\subsection{CRC-Polar concatenation} \label{sec:AddCRC}
We have considered two approaches for incorporating CRC in the proposed algorithm. In the first, we can think of the CRC as additional PC constraints and add them to the PCM before applying the pruning procedure. Unfortunately, the PCM of the CRC is not sparse, and this degrades the effectiveness of the pruning and results in a larger matrix.
As an example, for $\cP\left(256, 134\right)$ ($\cP\left(512, 262\right)$, respectively) with CRC of length $6$, the pruned PCM has blocklengh $N'=533$ ($N'=1150$), compared to $N'=355$ ($N'=773$) for a plain polar PCM.

Therefore, we used an alternative approach, where we add the additional CRC constraints after the pruning process. Since the CRC is a high rate code, this only adds a small number of equations to the PCM.
First, we need to obtain the CRC constraints in terms of the codeword (last $N$ columns of the pruned polar PCM). Denote by $\bH_{\crc}$ the PCM of the CRC, i.e., for every information word $\bu$ (including the frozen bits), we have $\bH_{\crc}\cdot \bu = \mathbf{0}$. It is known from \cite{LP_Polar_decoding} that $\bu^T=\bc^T\cdot \bG_N$ where $\bG_N$ is the $N \times N$ full generator matrix of the polar code. Hence, the CRC constraints can be expressed in terms of the codeword as $\mathbf{0}=\bH_{\crc}\cdot \bu = \bH_{\crc}\cdot \bG_N^T\cdot \bc$. That is, the CRC constraints that we add to the pruned polar PCM constraints are $\bH_{\crc}\bG_N^T$. In order to decrease the density (number of ones) of $\bH_{\crc}\bG_N^T$, we used the greedy algorithm proposed in \cite{PolarBPCRCBrink}. This algorithm considers the Hamming weight of the sum of every two rows, $i$ and $j$. If this sum is smaller than that of either row $i$ or row $j$, then the algorithm replaces the row (either $i$ or $j$) with larger Hamming weight, with this sum.  

\subsection{Simulation results}\label{sec:sim_res}
We used numerical simulations to evaluate the error rate and computational complexity of the proposed algorithm and compare it to (CRC-aided) SCL. In all simulations, we used the standard concatenated CRC-Polar scheme with a total rate of $1/2$ and CRC of length $6$. Fig. \ref{fig:BEC_res} shows the BER / FER performance and the mean number of reference variables, $n_r$, and remaining equations, $n_e$ (when the codeword was successfully decoded by the BP algorithm, which is applied in the first stage, $n_r = n_e = 0$). Whenever the diagonal cannot be further extended, we choose a single reference variable, that is $n'_r=1$. We considered two methods for selecting the reference variable. The first chooses the reference variable at random among all remaining unknown CVNs. The second, following \cite[Method C]{LDPC_ML}, chooses the unknown variable from a PC with the smallest number of remaining unknown variables. The second approach was slightly better and is hence the one used in Fig. \ref{fig:BEC_res}.
Since our decoder is the ML decoder, it achieves the best FER performance compared to all other decoders. As can be seen in Fig. \ref{fig:BEC_res}, our ML decoder also achieves the best BER performance. In the last row of Fig. \ref{fig:BEC_res} we can see that the number of required reference variables, $n_r$, and the number of remaining equations, $n_e$, are small. For example, for a codeword of length $N=512$ and $\epsilon \le 0.37$, the average number of required reference variables is less than $0.1\%$ of the blocklength.
\begin{figure}
	\centering
	\begin{subfigure}[b]{.3\linewidth}
		\includegraphics[width=\linewidth]{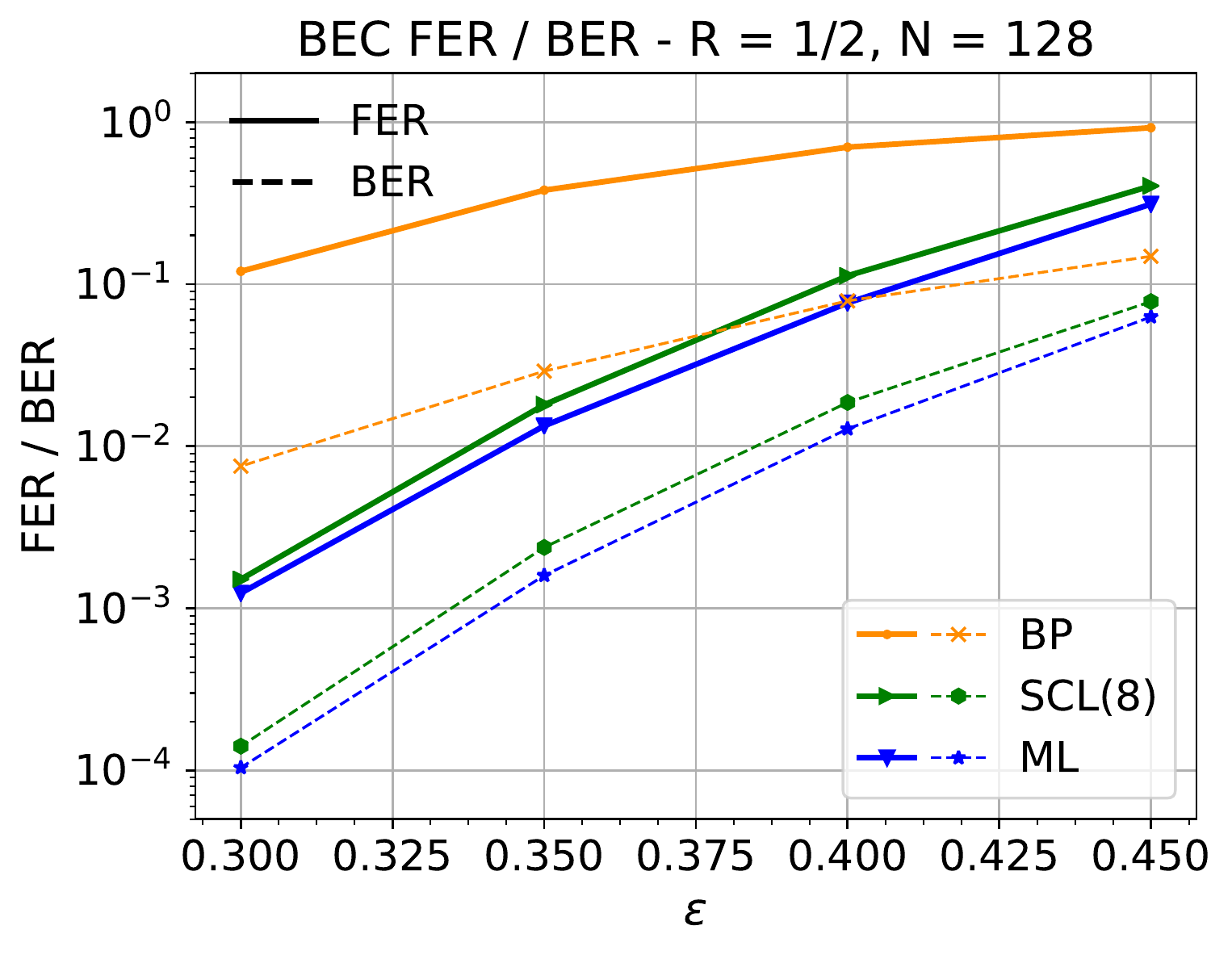}
	\end{subfigure}
	\hspace{0.006\textwidth}
	\begin{subfigure}[b]{.3\linewidth}
		\includegraphics[width=\linewidth]{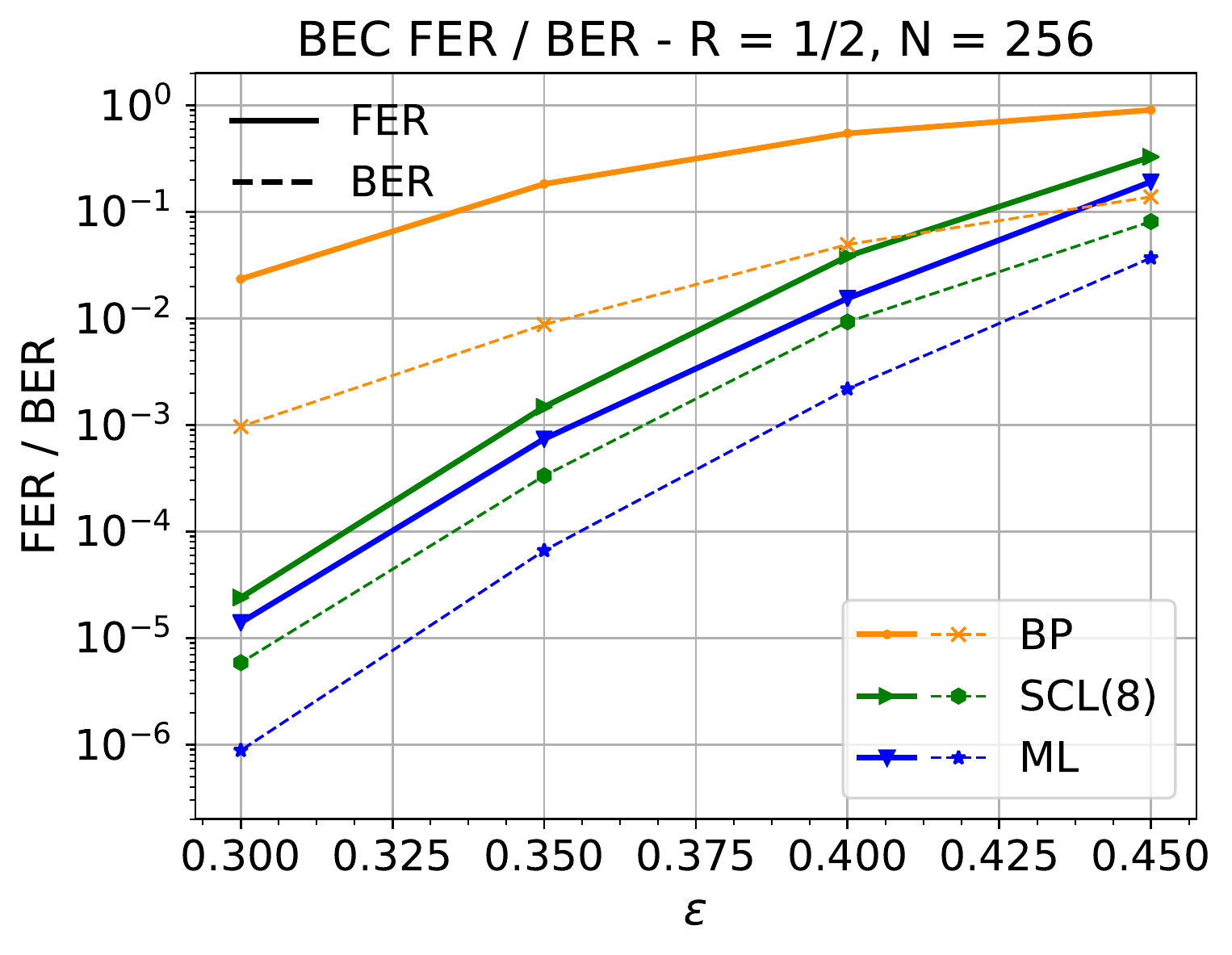}
	\end{subfigure}
	\hspace{0.006\textwidth}
	\begin{subfigure}[b]{.3\linewidth}
		\includegraphics[width=\linewidth]{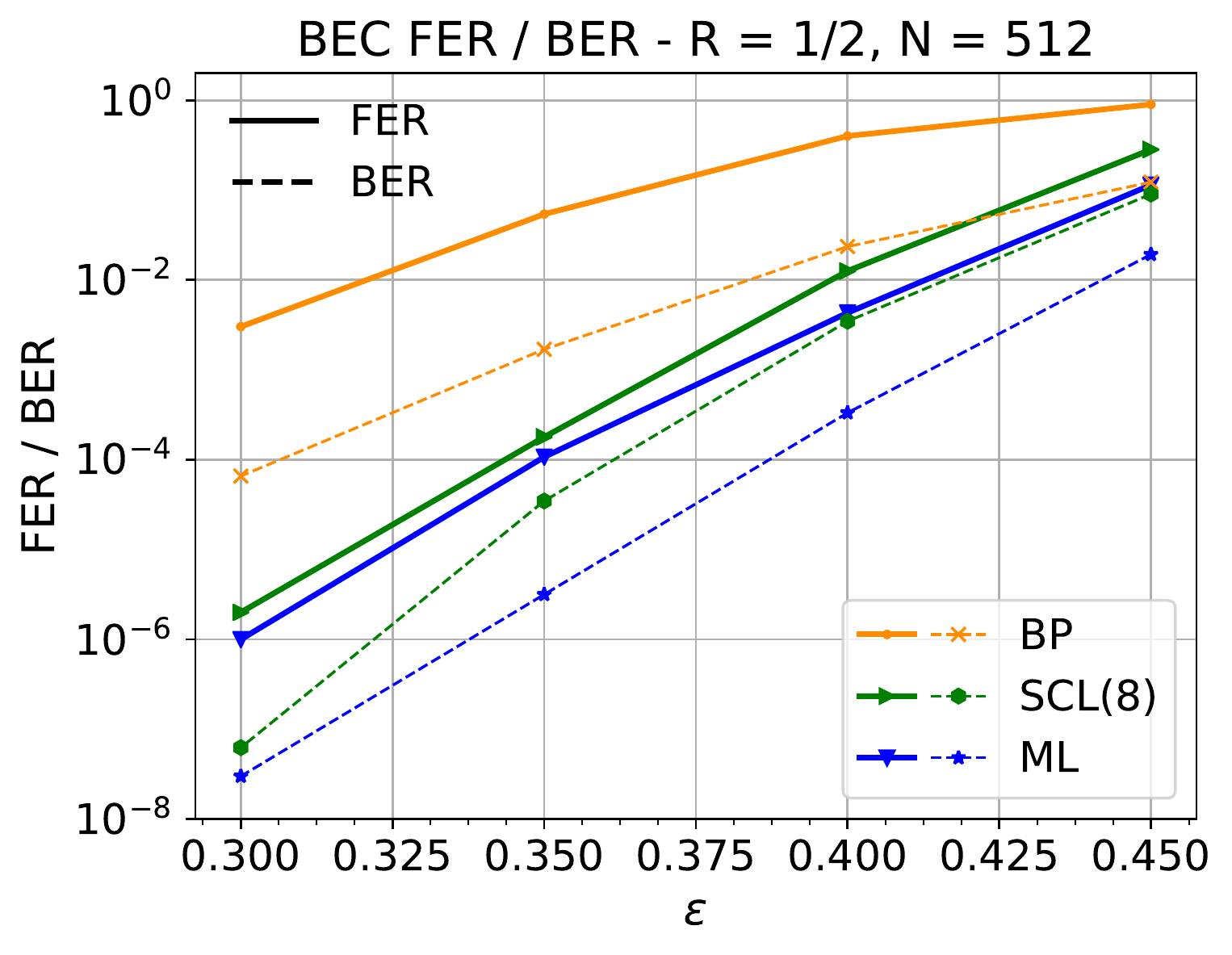}
	\end{subfigure}
	
	\vspace{2ex}
	\hspace{0.002\textwidth}
	\begin{subfigure}[b]{.3\linewidth}
		\includegraphics[width=\linewidth]{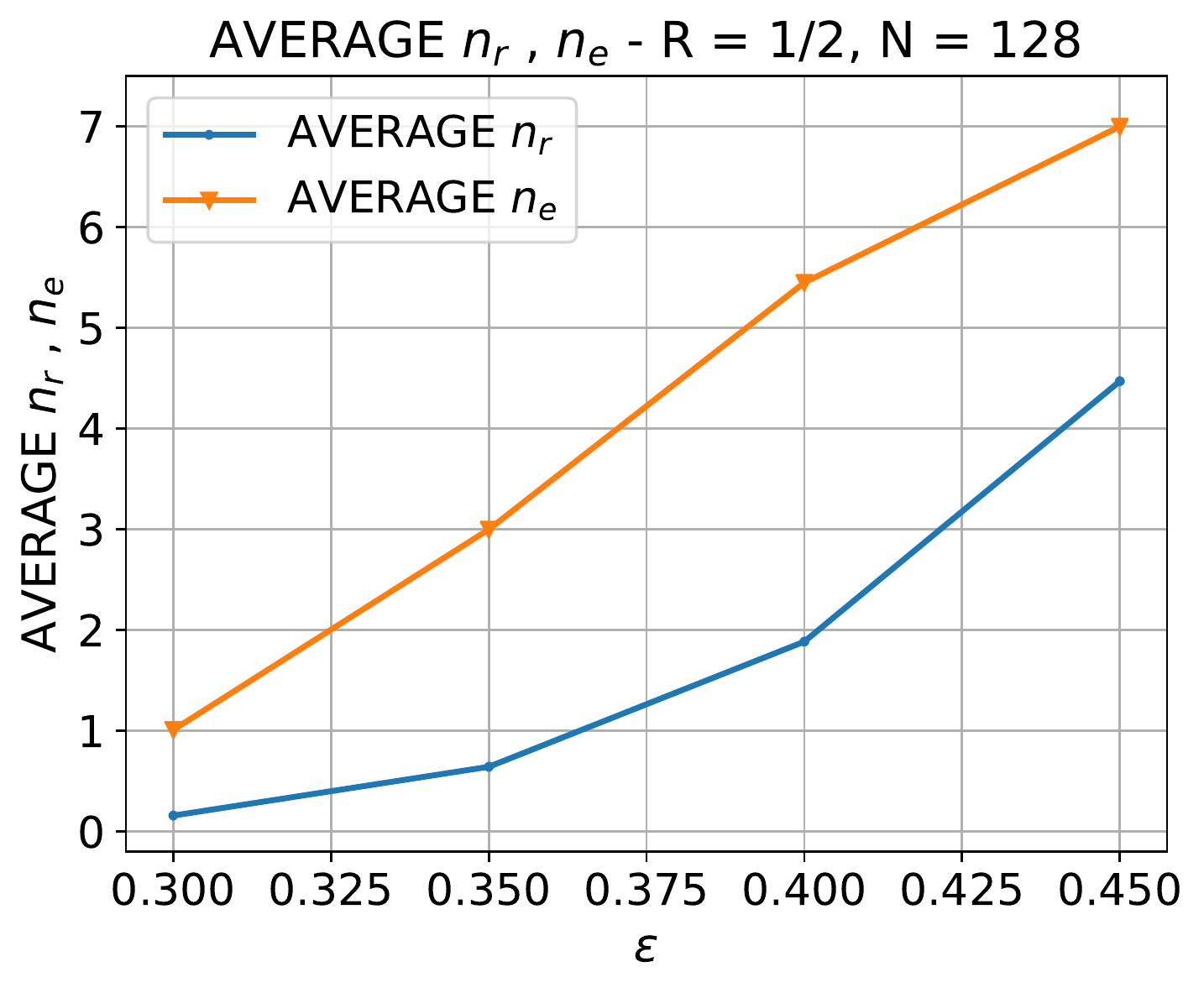}
	\end{subfigure}
	\hspace{0.01\textwidth}
	\begin{subfigure}[b]{.3\linewidth}
		\includegraphics[width=\linewidth]{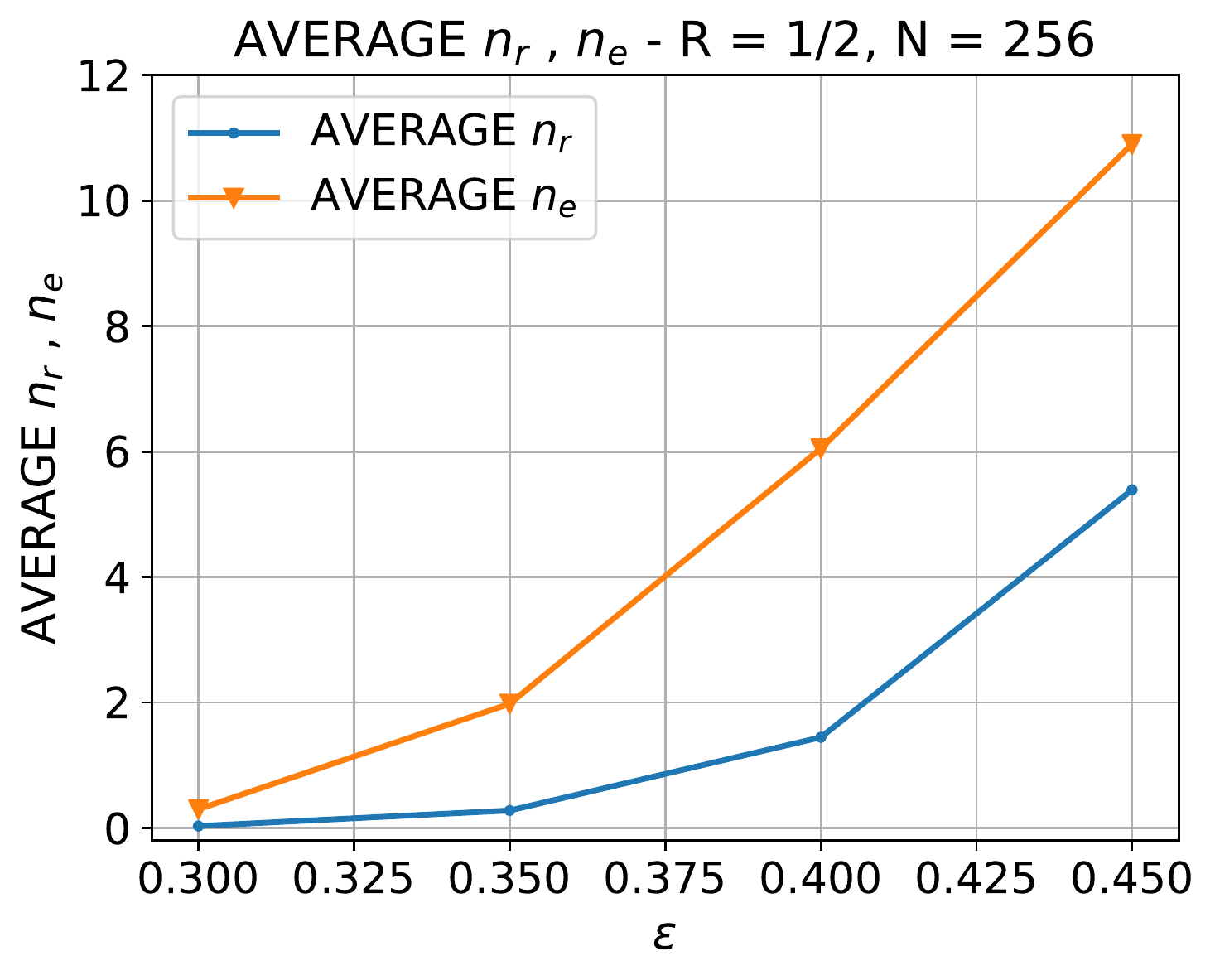}
	\end{subfigure}
	\hspace{0.008\textwidth}
	\begin{subfigure}[b]{.3\linewidth}
		\includegraphics[width=\linewidth]{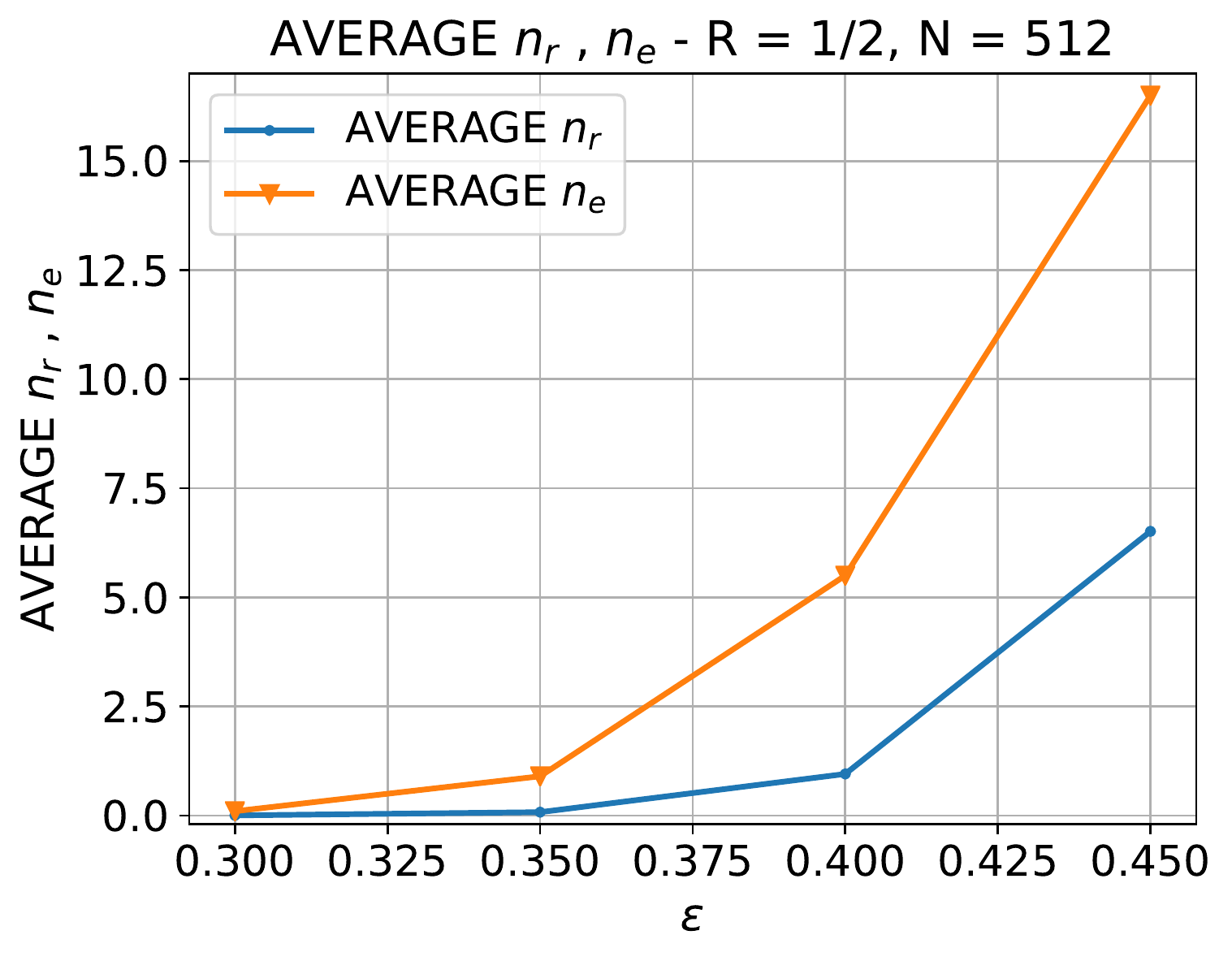}
	\end{subfigure}
	
	\caption{BER, FER, and average $n_r$ and $n_e$, for CRC concatenated polar codes with rate $1/2$ and different blocklengths over the BEC.}
	\label{fig:BEC_res}
\end{figure}

\subsection{Parallel implementation} \label{sec:parImp}
For efficient parallel implementation, we suggest modifying stage 2 in the decoding algorithm in Section \ref{sec:EffML}. This change simplifies stages 3 and 4.
Recall that the goal of stage 2 is to obtain the PCM shown in Fig. \ref{fig:FinalPCM}. The revised algorithm further specifies $\bH^{(1,3)}=I$ ($n_u \times n_u$ diagonal matrix) and $\bH^{(2,3)}=0$.
For that purpose we modify stage 2 by adding the following row elimination step after each diagonal extension step.
Consider for example Fig. \ref{fig:diag_exten} that describes the PCM after the first diagonal extension, where the number of reference variables is $n'_r$. Suppose that the first diagonal extension step found $l_1$ rows with a single unknown variable node, so that after the first diagonal extension step, the PCM is shown in Fig. \ref{fig:diag_exten} with $n=0$, $l=l_1$, $\bR=\bI$ and $n'_r$ reference variables. At this point, we add the row elimination step: We use row additions (XORs) with the $l=l_1$ rows of $\bR=\bI$, in order to zero out the sub-matrix $\bB$ in Fig. \ref{fig:diag_exten} (i.e., after applying these row additions, $\bB=\mathbf{0}$). The point is that these row additions can be executed in parallel.
The same row elimination step is added after each diagonal extension step.

Stages 3 and 4 simplify considerably due to the modification in stage 2.
Using $\bH^{(1,3)} = \bI$ in \eqref{eq:mat_eq_s1_s2} yields \eqref{eq:u_Ar_b} for $\bA=\bH^{(1,2)}$, $\ba=\bH^{(1,1)}\bd$.
Using $\bH^{(2,3)}=0$ simplifies \eqref{eq:linear_system_Ar_b} to
$\bH^{(2,2)} \br = \bs^{(2)}$.
To solve this system, which is typically small ($n_e \times n_r$), we can use the parallel Gaussian elimination over GF(2) algorithm proposed in \cite{FastGausElim} with time complexity $\mathcal{O}(M)$ for a random $M\times M'$ matrix.
Our approach in the diagonalization process is somewhat similar, but since our matrix is sparse, we can find several weight one rows simultaneously, and we can use this to pivot in parallel more than one column at a time.

Although $n'_r=1$ minimizes the total number of reference variables, $n_r$, used, for efficient parallel implementation of stage 2 it is beneficial to use $n'_r>1$.

\section{polar code CBPL-OSD decoding} \label{sec:OSD}
In this section, we present an efficient CBPL-OSD decoder for polar codes transmitted over the AWGNC. 

Consider a linear block code $\cC(N, K)$. Suppose that the BPSK modulated codeword, $\text{BPSK}(\bc)$, is transmitted over an AWGNC. The received signal is given by,
\bre
\by = \text{BPSK}(\bc) + \bn
\label{eq:AWGNC}
\ere
where $\text{BPSK}(\bc)\triangleq (-1)^{\bc} \triangleq((-1)^{c_1},\cdots,(-1)^{c_{N}})$, and
where the elements of $\bn$ are i.i.d Gaussian, $n_i\sim\cN(0,\sigma^2)$.
The uncoded LLRs of the codeword's bits, $\ell^{\ch}_i$, based on the respective channel output, $y_i$, are given by
\begin{equation}
\ell^{\ch}_i =
\ln \frac{\Pr(c_i=0 \mid y_i)}{\Pr(c_i=1 \mid y_i)} = 2\frac{y_i}{\sigma^2}
\label{eq:LLRxi}
\end{equation}
for $i=1,\ldots,N$, where $\sigma^2 = N_0/2$.
Also denote by $\bl^{\ch} = (\ell^{\ch}_1,\ldots,\ell^{\ch}_N)$.

We start by providing a brief review on OSD \cite{OSD} applied in this setup.

\subsection{Ordered statistics decoder (OSD)} \label{sec:osd_background}
Denote the received LLRs vector by $\bl = \bl^{\ch}$. We define by $|\ell_i|$ the reliability of the $i$'th received symbol.
The OSD algorithm consists of two main parts \cite{OSD}: Finding the most reliable independent basis (MRIB) of the code with respect to the values in $\bl$, and a reprocessing stage.
The process of finding the MRIB is described in detail in \cite{OSD} using a generator matrix of the code. This process can be equivalently described using a full row rank PCM of the code, $\bH$, of size $(N-K) \times N$ as follows.
We first reorder the components of $\bl$ in decreasing order of reliability. Denote the resulting vector by $\bl^1$ and the corresponding permutation by $\lambda_1(\cdot)$, i.e.,
\bre
\bl^1 = (\ell_1^1, \ell_2^1, \ldots, \ell_N^1) = \lambda_1(\bl)
\ere
where $|\ell_1^1| \ge |\ell_2^1| \ge \ldots \ge  |\ell_N^1|$. Similarly we permute the columns of $\bH$ using $\lambda_1$ to obtain $\bH^{(1)} = \lambda_1(\bH)$. Next, starting from the rightmost column of $\bH^{(1)}$, we find the first from the right $N-K$ linearly independent columns. We denote the remaining $K$ columns as the most reliable independent basis (MRIB). We now permute the columns of $\bH^{(1)}$ such that it's first $K$ columns are the MRIB and denote the resulting matrix by $\bH^{(2)}$ and the corresponding permutation by $\lambda_2(\cdot)$. Finally, we bring the PCM to the systematic form
\begin{equation}
\tilde{\bH}=[\ \textbf{A}\ |\ \textbf{I}_{N-K}\ ]
\label{eq:OSDH}
\end{equation}
using elementary row operations, where $\textbf{A}$ is a $(N-K)\times K$ matrix, and
$\textbf{I}_{N-K}$ is the $(N-K)\times (N-K)$ identity matrix.
Note that the aforementioned steps have not altered the code except for permuting its codebits order, such that $\tilde{\cC} = \lambda(\cC) = \lambda_2(\lambda_1(\cC))$, where $\tilde{\cC}$ is the code corresponding to the PCM $\tilde{\bH}$. Also, note that in practice we typically combine the step of finding the MRIB with the step of bringing the PCM to a systematic form.
Denote by $\tilde{\bl}=\lambda(\bl)$ ($\tilde{\by} = \lambda(\by)$, respectively) the permutation of the LLRs vector, $\bl$ (channel output vector, $\by$), using the same permutation $\lambda$.

We can now start the reprocessing stage of the OSD algorithm. In OSD of order $q$, denoted OSD($q$), we generate a list of $\sum_{k=0}^{q} {K \choose k}$ candidate codewords using reprocessing as follows.
We split each codeword, $\bc \in \tilde{\cC}$, into two parts $\bc = [\bc^{(1)}, \bc^{(2)}]$, where $\bc^{(1)}$ is of length $K$ and $\bc^{(2)}$ is of length $N-K$. Then, by the PCM constraints $\tilde{\bH} \bc = \mathbf{0}$ and \eqref{eq:OSDH}, we have,
$\bc^{(2)}  = \bA \bc^{(1)}$.
Denote the $K$ hard decoded bits corresponding to the MRIB by $\bc_0^{(1)}$,
\begin{equation}
c^{(1)}_{0,j}=\begin{cases}
0, & \text{$\tilde{\ell}_{j}\geq 0$}\\
1, & \text{else}
\end{cases} \quad , \quad 1\leq j\leq K \: .
\label{eq:c01j}
\end{equation}
We enumerate all binary vectors, $\hat{\be}_i$, of Hamming weight at most $q$ (error patterns), and generate a list of codewords $\bc_i = [\bc_i^{(1)}, \bc_i^{(2)}] \in \tilde{\cC}$ using
\begin{equation}
\bc_i^{(1)} = \bc_0^{(1)} + \hat{\be}_i,   \quad
\bc_i^{(2)} = \bA \bc_i^{(1)} \label{eq:ci_1_2}.
\end{equation}
We select the (permuted) codeword, $\bc_i$, for which the distance between $\text{BPSK}(\bc_i)$ and $\tilde{\by}$ is minimal, and inversely permute it by $\lambda^{-1}$ to obtain the OSD($q$) estimate of the transmitted codeword.
Intuitively, since $\bc_i^{(1)}$ comprises of the most reliably received bits (more precisely, the MRIB bits), we expect a small number of errors in the hard decoded bits $\bc_0^{(1)}$. As a result, the above list will contain the true transmitted codeword with high probability. Obviously, OSD($q$) with $q=K$ is ML decoding.

\subsection{CBP-OSD and CBPL-OSD decoding of polar codes} \label{sec:cbplosd}
In \cite{BPOSD} it was proposed to combine BP decoding of LDPC codes with OSD by first running the BP decoder and then using the soft decoding output of the codebits in order to determine the MRIB (i.e., rather than sorting the codebits by the absolute values of their uncoded LLRs, they are sorted by the absolute values of their soft BP decoded LLRs).
We take a similar approach for BP decoding of polar codes. We start by describing the CBP-OSD decoder. 
Instead of using the channel output LLRs as the input to the OSD algorithm (i.e., $\bl = \bl^{\ch}$), we propose to initially run a CBP decoder on the CRC augmented polar code FG with some stopping criterion. If the decoding process has not terminated (i.e., the stopping criterion was not satisfied), then we apply OSD($q$) using the output soft decoded LLRs of the codebits. The motivation is that the decoder's soft output is typically much more reliable than the uncoded LLRs (e.g., its hard decoding has lower BER). Moreover, for the codebits with the highest reliability, the BER is even lower. As a result, applying OSD of low order is typically sufficient to improve the decoding error rate performance of CBP significantly. In fact, even when using $q=1$, the resulting CBP-OSD decoder has a significantly lower error rate compared to plain CBP decoding.
Unfortunately, the computational complexity of the Gaussian elimination required by OSD or CBP-OSD to bring $\bH$ into the form \eqref{eq:OSDH} may pose a limitation. To reduce the complexity of this stage, we use a variant of the algorithm presented in Section \ref{sec:EffML} for efficient ML decoding of polar codes over the BEC, to bring the sparse pruned $(N'-K) \times N'$ PCM of the polar code to the form \eqref{eq:OSDH}.

We now describe the CBP-OSD algorithm in detail.
\begin{enumerate}
	\item Apply CBP decoding on the polar code with the following stopping criterion. Denoting by $\hat{\bu},\hat{\bc}$ the hard decoding values of the CBP-decoded polar message and codebits, respectively, terminate if $\hat{\bc}^T=\hat{\bu}^T\cdot\bG_N$, and in addition	
	$\hat{\bu}_{\cA}$ satisfies all CRC constraints, i.e., $\bH_{\crc}\hat{\bu}_{\cA}=\mathbf{0}$.	
	As soon as the early termination criterion is satisfied, we return the decoded word and exit. If early termination was not achieved after a certain number, $I_{\max}$, of iterations, we proceed with the following stages.
	
	\item Sort the obtained output LLRs in decreasing order of reliability. Permute the columns of the pruned PCM using the same order.
	
	\item Mark the first $K$ columns as `fixed'. Now, perform stage 2 in the algorithm in Section \ref{sec:parImp} (we could have used instead the first variant of the algorithm, presented in the beginning of Section \ref{sec:EffML}, stage 2), with the $K$ fixed columns instead of the $n_d$ (BP) decoded variables that we had in Section \ref{sec:parImp} and with $n_c=0$. We start by performing diagonal extension steps, then choose reference variables and perform additional diagonal extension steps iteratively. In the end of that stage we obtain the PCM in Fig. \ref{fig:PCMAfterDiagOSD}. We examined two criteria for choosing new reference variables, as described below.
	
	\item Permute the rows of the PCM that we obtained by moving the $n_r$ rows at the bottom to the top of the PCM as can be seen in Fig. \ref{fig:PCMAfterDiagOSDReorderd}.
	\begin{figure}[h]
		\centering
		\begin{subfigure}[t]{.4\linewidth}
			\hspace*{-1cm}\includegraphics[width=\linewidth]{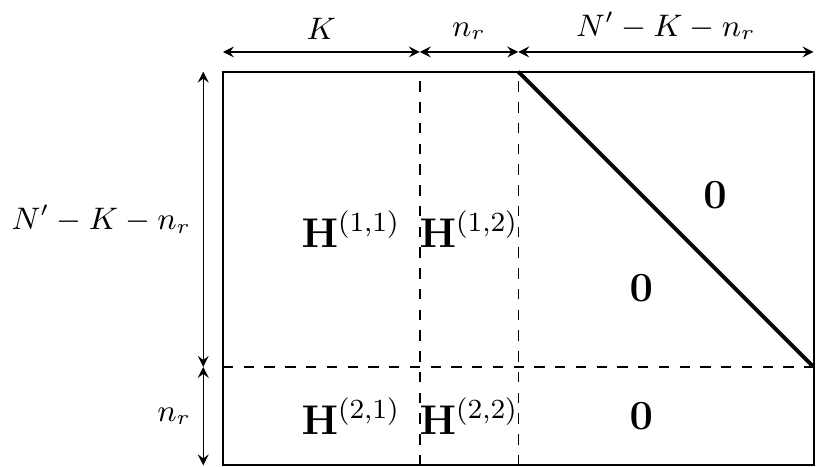}
			\caption{The PCM after choosing reference variables and diagonalizing.}
			\label{fig:PCMAfterDiagOSD}
		\end{subfigure}
		\begin{subfigure}[t]{.4\linewidth}
			\hspace*{-0.7cm}\includegraphics[width=\linewidth]{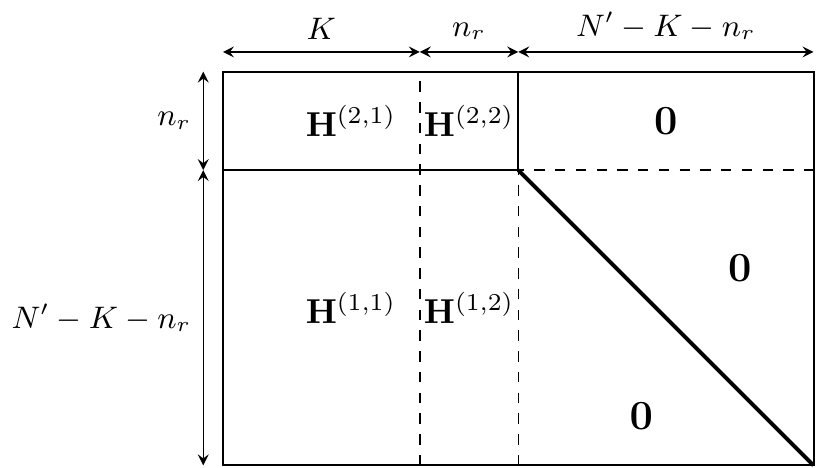}
			\caption{The row permuted PCM.}
			\label{fig:PCMAfterDiagOSDReorderd}
		\end{subfigure}
		\caption{The PCM in stage 3 of our OSD algorithm. The bold diagonals are filled with ones.}
		\label{fig:OSDPCM}
	\end{figure}
	Then, apply Gaussian elimination on the $n_r \times (K+n_r)$ upper left part of the matrix in Fig. \ref{fig:PCMAfterDiagOSDReorderd}, denoted by $\left[ \bH^{(2,1)} \: | \: \bH^{(2,2)} \right]$, with possible column permutations to bring the result to the form $\left[ \tilde{\bH}^{(2,1)} \: | \: \bI_{n_r} \right]$, where $\tilde{\bH}^{(2,1)}$ is $n_r \times K$ and $\bI_{n_r}$ is the identity matrix of order $n_r$. We claim that this is always possible. To prove this claim first recall that by Proposition \ref{prop:pruned_PCM_full_rank}, the pruned PCM has full row rank. Also, stage 2 in the algorithm in Section \ref{sec:parImp} does not change the row rank. Hence, the resulting PCM shown in Fig. \ref{fig:PCMAfterDiagOSD} also has full row rank, and, by the structure of this matrix, $\left[ \bH^{(2,1)} \: | \: \bH^{(2,2)} \right]$ must also have full row rank, $n_r$. This proves our claim.
	We implement the Gaussian elimination together with column permutation on $\left[ \bH^{(2,1)} \: | \: \bH^{(2,2)} \right]$ such that the $K$ leftmost columns (corresponding to $\tilde{\bH}^{(2,1)}$) are those associated with the more reliable BP decoded variables. By imposing this, the first $K$ columns of the final PCM will be the MRIB provided that our criterion for choosing reference variables in stage 3 is to always select the remaining codeword variables with the highest reliabilities. We have shown how to bring the PCM to the form \eqref{eq:OSDH}, but its dimensions are $(N'-K)\times N'$ rather than $(N-K) \times N$ (where $N'\ge N$) since it includes extra variables beyond the $N$ codeword variables (variable nodes in the FG description of the polar code that were not pruned by the pruning algorithm). However, in the reprocessing stage of the OSD we only care about codeword variables. Hence, at this point we simply erase the $N'-N$ rows and columns, corresponding to variable nodes which are not codeword variables, from the PCM (more precisely, if column $K+i$, for $i\in\{1,\ldots,N'-K\}$, is not a codeword variable node, we erase it together with row $i$) and obtain the final PCM \eqref{eq:OSDH} with dimensions $(N-K) \times N$.
	
	\item To generate the list of candidate codewords required by the reprocessing stage of the OSD algorithm, we follow the description in Section \ref{sec:osd_background}, \eqref{eq:c01j}-\eqref{eq:ci_1_2}. For the case $q=1$ we can rewrite this as follows. Let $\bh_0 \defined \mathbf{0}$ and let $\bh_i$, $i=1,\ldots,K$, be the $i$'th column of the final PCM in \eqref{eq:OSDH}. Also, $\be_0\defined \mathbf{0}$ and $\be_i$, $i=1,\ldots,K$, is the $i$-th unity vector of length $K$. Then, by \eqref{eq:ci_1_2}, $\bc_i = [\bc_i^{(1)},\bc_i^{(2)}]$, $i=0,1,\ldots,K$, is given by,
	\begin{align}
	\bc_i^{(1)} = \bc_0^{(1)} + \be_i \quad , \quad
	\bc_i^{(2)} = {\bc}_0^{(2)} + \bh_i
	\label{eq:ci_1_2A}
	\end{align}
	where $\bc_0^{(1)}$ is given by \eqref{eq:c01j} and ${\bc}_0^{(2)} = \bA {\bc}_0^{(1)}$.
	
	\item Finally, the decoder outputs $\lambda^{-1}({\bc}_{i_0})$ for
	\bre
	i_0 = \argmin_{0\le i \le K} ||\text{BPSK}(\bc_i) - \tilde{\by} ||^2
	\label{eq:i_0}
	\ere
	Denoting by $c_{i,l}$ ($c^{(1)}_{0,l}$, $c^{(2)}_{0,l}$, $\tilde{y}_l$, $e_{i,l}$, $h_{i,l}$, respectively) the $l$'th component of $\bc_{i}$ ($\bc^{(1)}_0$, $\bc^{(2)}_0$, $\tilde{\by}$, $\be_{i}$, $\bh_i$), we have
	\begin{align}
	i_0 &= \argmax_{0\leq i\leq K}\sum_{l=1}^{N} (-1)^{c_{i,l}} \tilde{y}_l
	= \argmax_{0\leq i\leq K} \left\{ 
	\sum_{l=1}^{K} (-1)^{c^{(1)}_{0,l}} (-1)^{e_{i,l}} \tilde{y}_l +
	\sum_{l=1}^{N-K} (-1)^{c^{(2)}_{0,l}} (-1)^{h_{i,l}} \tilde{y}_{K+l} \right\}
	\\
	&= \argmax_{0\leq i\leq K} \left\{ 
	-\sum_{l=1}^{K} (-1)^{c^{(1)}_{0,l}} \tilde{y}_l +
	\sum_{l=1}^{K} (-1)^{c^{(1)}_{0,l}} (-1)^{e_{i,l}} \tilde{y}_l +
	\sum_{l=1}^{N-K} (-1)^{c^{(2)}_{0,l}} (-1)^{h_{i,l}} \tilde{y}_{K+l}
	\right\}
	\\
	&= \argmax_{0\leq i\leq K} \left\{
	-2\cdot (-1)^{c^{(1)}_{0,i}} \tilde{y}_i + \sum_{l=1}^{N-K}s_l\cdot(-1)^{h_{i,l}}
	\right\}
	\label{eq:minL2}
	\end{align}
	where $\tilde{y}_0 \defined 0$ and
	\bre
	s_l \defined \tilde{y}_{K+l}\cdot(-1)^{c^{(2)}_{0,l}}, \quad l=1,2,\ldots,N-K \: .
	\label{eq:sdef}
	\ere
	To verify the last equality in \eqref{eq:minL2} note that $e_{i,l}=1$ if $i=l$, otherwise $e_{i,l}=0$.
	Using \eqref{eq:minL2}, the reprocessing, which can be implemented in parallel, requires $(N-K)(K+1)$ additions / subtractions.
\end{enumerate}

It should be noted that all the manipulations of the PCM required by our efficient Gaussian elimination algorithm are simple operations on binary data. To implement the Gaussian elimination required in stage 4 efficiently, we can use the parallel Gaussian elimination over GF(2) algorithm proposed in \cite{FastGausElim}. As will be shown in the simulations, $n_r$ is typically considerably smaller than $K$ and $N-K$, so that this Gaussian elimination is much less demanding compared to plain Gaussian elimination of a standard $K \times N$ polar generator matrix or a standard $(N-K) \times N$ polar PCM described in Section \ref{sec:osd_background}.

The same algorithm can be used for BPL-OSD or CBPL-OSD decoding by applying the same procedure of BP or CBP followed by OSD decoding (with efficient Gaussian elimination) on several FG permutations (or permuted inputs and outputs \cite{BPPermutedWarren}) in parallel.

\subsection{Lower comlexity CBP-OSD and CBPL-OSD} \label{sec:cbpllcosd}
We now suggest a lower complexity variant of the decoder in Section \ref{sec:cbplosd}.  
We first apply stages 1-3 of the decoder, thus obtaining the PCM in Fig. \ref{fig:PCMAfterDiagOSD}. We skip stage 4 that includes the Gaussian elimination of a matrix of dimensions $n_r \times (K + N_r)$ and proceed directly to the reprocessing part of OSD (stages 5 and 6). However, the reprocessing is modified as follows.
We compute the hard decoded bits, $\bc_0^{(1)}$, corresponding to the first $K+n_r$ variable nodes in the PCM in Fig. \ref{fig:PCMAfterDiagOSD} similarly to \eqref{eq:c01j} and enumerate over all $K+n_r+1$ error patterns, $\be_i$, of weight at most 1. For each vector, $\bc_0^{(1)}+\be_i$, that satisfies the $n_r$ parity check constraints at the bottom of that PCM, we calculate the remaining $N-K-n_r$ codeword variable nodes (as in stage 4 of the decoder in Section \ref{sec:cbplosd}, we erase from the PCM the $N'-N$ rows and columns corresponding to variable nodes which are not codeword variables). Similarly to \eqref{eq:i_0}, the decoded codeword is the one with minimum $\cL_2$ distance to $\tilde{\by}$.

\subsection{Partial higher order OSD reprocessing} \label{sec:pool}
So far we have focused on CBPL-OSD(1), as it has the lowest computational complexity. Yet, a lower error rate may be achieved by using CBPL-OSD($q$) for $q>1$. In this section, we propose approximating higher order reprocessing by performing it only on the bits that are associated with the least reliable LLRs. We demonstrate this approximation for CBPL-OSD(2), and show that it can be used to achieve almost the same performance as that of regular CBPL-OSD(2), while significantly reducing the reprocessing complexity.

Recalling \eqref{eq:ci_1_2A} and using the same notation as in Section \ref{sec:cbplosd}, let us further denote by $\bc_{i,j} = [\bc_{i,j}^{(1)},\bc_{i,j}^{(2)}]$ the codeword after flipping the $i$'th and $j$'th bits in $\bc_0^{(1)}$, such that
\begin{equation}
\bc_{i,j}^{(1)} = \bc_0^{(1)} + \be_i + \be_j , \quad
\bc_{i,j}^{(2)} = {\bc}_0^{(2)} + \bh_i + \bh_j \: .
\label{eq:cij_1_2}
\end{equation}
The main bottleneck of order-2 reprocessing is the search for the codeword $\bc_{i,j}$ with BPSK closest to the permuted channel output $\tilde{\by}$ in terms of Euclidean distance,
\begin{equation}
\argmin_{1\leq i<j\leq K}\norm{\text{BPSK}(\bc_{i,j}) - \tilde{\by}}^2 \: .
\label{eq:minL2A}
\end{equation}
To determine the final CBPL-OSD(2) estimate, the result of the search \eqref{eq:minL2A} needs to be compared to the CBPL-OSD(1) estimate in \eqref{eq:minL2} by minimum Euclidean distance to $\tilde{\by}$.

In \eqref{eq:minL2A} we search over $\binom{K}{2}=[K(K-1)] / 2$ pairs of indices $(i,j)$. Similarly to \eqref{eq:minL2}, we can rewrite \eqref{eq:minL2A} as,
\begin{align}
\lefteqn{\argmin_{1\leq i<j\leq K}\norm{\text{BPSK}(\bc_{i,j}) - \tilde{\by}}^2 =}\\
&\qquad
\argmax_{1\leq i<j\leq K} \left\{ -2\cdot (-1)^{c^{(1)}_{0,i}} \tilde{y}_i
-2\cdot (-1)^{c^{(1)}_{0,j}} \tilde{y}_j 
+ \sum_{l=1}^{N-K} s_l\cdot(-1)^{h_{i,l}}\cdot(-1)^{h_{j,l}} \right\} \: .
\label{eq:minL2B}
\end{align}
Using the RHS of \eqref{eq:minL2B}, the complexity of this search, and thus of order-2 reprocessing, is about $(N-K) K^2 / 2$ additions.

The RHS of \eqref{eq:minL2B} can be expressed using matrix multiplication (that only requires additions / subtractions, and can be implemented efficiently in parallel) as follows.
First define a $K\times (N-K)$ matrix $\bA$ such that the $i$'th row, $i=1,\ldots,K$, of $\bA$ is $\bs \cdot (-1)^{\bh_i}$ (element-wise multiplication and exponentiation). Next, define an $(N-K) \times K$ matrix $\bB$ such that the $j$'th column, $j=1,\ldots,K$, of $\bB$ is $(-1)^{\bh_j}$.
Finally define the $K\times K$ matrix $\bD$ by $\bD = \bA \cdot \bB$.
Then the RHS of \eqref{eq:minL2B} can be written as
\begin{equation}
\argmax_{1 \le i < j \le K} \left\{ -2\cdot (-1)^{c^{(1)}_{0,i}} \tilde{y}_i
-2\cdot (-1)^{c^{(1)}_{0,j}} \tilde{y}_j + d_{i,j} \right\} \: .
\end{equation}
A similar formulation using matrix multiplication applies to the partial reprocessing method that we now suggest.

Recall the permuted (by $\lambda$) LLRs $(\tilde{\ell}_1,\ldots,\tilde{\ell}_K)$, corresponding to the MRIB, defined in Section \ref{sec:osd_background}. To reduce the complexity we propose the following approximation to the search in \eqref{eq:minL2B}. For some integer $M < [K(K-1)] / 2$, we perform this search only on the $M$ pairs of indices $(i,j)$ with the lowest values of $|\tilde{\ell}_i|+|\tilde{\ell}_j|$. That is, we only consider the $M$ pairs $(i,j)$ which are associated with the least reliable LLRs (LLRs with the lowest absolute values), in the maximization \eqref{eq:minL2B}. Since $|\tilde{\ell}_i| > |\tilde{\ell}_j|$ for $1\leq i < j \leq K$, this is implemented by enumerating over the indices $(i,j)$ during reprocessing in the following decreasing order: $i=K-1,K-2,\ldots$ in the outer loop and $j=K,K-1,\ldots,i+1$ in the inner loop, until we have exhausted $M$ such pairs. This approximation can be straightforwardly extended for higher orders. In the case of order 2, it reduces the complexity of the search over the pairs from about $(N-K) K^2 / 2$ to about $(N-K) M$ additions / subtractions. 

\subsection{Simulation results}
We present simulation results using our efficient CBPL-OSD decoder. We used a CRC-polar code of total rate $1/2$, and CRC of length $6$.
The maximal number of iterations for all BP decoders was set to $I_{\max}=100$. We start incorporating CRC information after $I_{\thr}=10$ iterations.
We compared the following decoders,
\begin{enumerate}
	\item (CRC-aided) SCL(8) with list size $L=8$.
	\item CBPL(6) \cite{BPL,PolarBPCRCBrink} with list size $L=6$ that includes all the permutations of the final $3$ stages of the FG.
	\item CBPL(6)-OSD(1), our proposed decoder described above with list size $L=6$.
	\item CBPL(6)-LCOSD(1), our lower complexity decoder described in Section \ref{sec:cbpllcosd}.
	\item CBPL(6)-OSD(2), similar to CBPL(6)-OSD(1), but the OSD now has order 2.
	\item CBPL(6)-POSD(2,$1/4$), similar to CBPL(6)-OSD(2), but now with partial order 2 reprocessing as described in Section \ref{sec:pool}, using only a quarter of the number of pairs compared to CBPL(6)-OSD(2).
\end{enumerate}
We simulated the decoding process for code blocklengths $128$, $256$, and $512$.
Fig. \ref{fig:OSDStats} presents the BER, FER and average total number of reference variables for the case where a single reference variable is chosen each time the diagonal cannot be further extended, i.e., $n'_r=1$.
We examined two methods for choosing the reference variables. The simplest approach is to choose the CVN with the highest reliability from the remaining unknown CVNs. The second approach, which is a variant of \cite[Method C]{LDPC_ML}, is to consider some PC with the smallest number of unknown variables (must be at least two since otherwise, we could have continued the diagonal extension) with at least one unknown CVN, and select the unknown CVN with the highest reliability as the reference variable. The results in Fig. \ref{fig:OSDStats} were obtained using the second approach, which was slightly better in terms of the average number of reference variables compared to the first one. The BER and FER performances of both methods were essentially the same.
Note that the average total number of reference variables was computed only over those cases where the early termination criterion of CBPL did not occur (i.e., only when OSD was required). \emph{These averages are significantly lower when all cases are included.}
It can be seen that the total number of reference variables is considerably smaller than $K$ and $N-K$.
Comparing with Gaussian elimination of a dense generator matrix of size $K \times N$, or a dense PCM of size $(N-K) \times N$, which is required by standard OSD, we observe a significant improvement in computational complexity.

We also observe that the lower complexity (compared to CBPL(6)-OSD(1)) CBPL(6)-LCOSD(1) decoder still provides a significant reduction in the error rate compared to plain CBPL(6).
Finally, the error rate of CBPL(6)-POSD(2,$1/4$) is very similar to that of CBPL(6)-OSD(2), although it only requires a quarter of the reprocessing complexity.

In addition to the experiment results shown in Fig. \ref{fig:OSDStats}, we have also compared CBPL(6)-OSD(1) to plain CBPL($L$) with a large value of $L$. Our results show that even for $L=64$, CBPL(6)-OSD(1) has a lower error rate (the blocklength was $N=256$).

\begin{figure}
	\centering
	\begin{subfigure}[b]{.3\linewidth}
		\includegraphics[width=\linewidth]{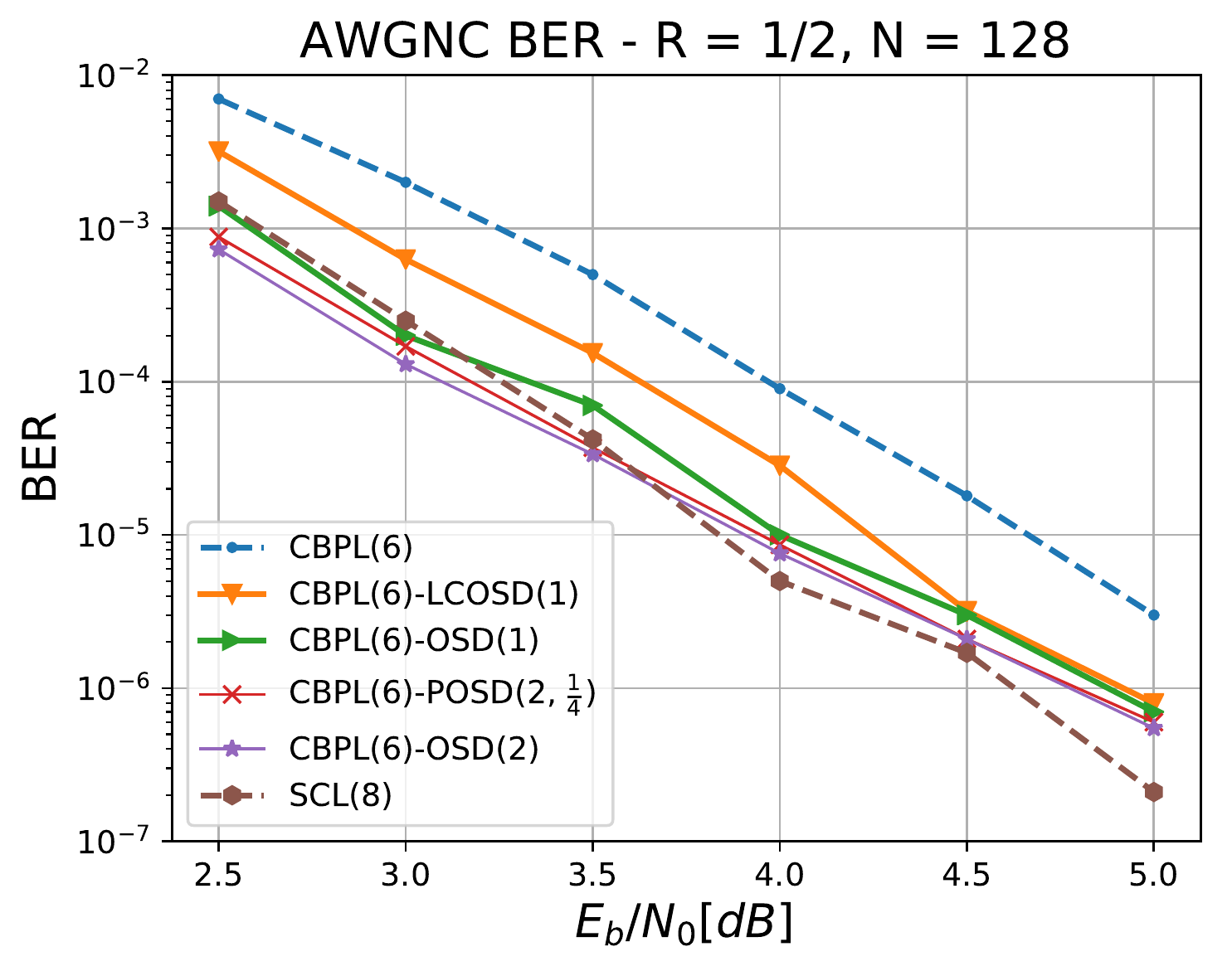}
	\end{subfigure}
	\hspace{0.006\textwidth}
	\begin{subfigure}[b]{.3\linewidth}
		\includegraphics[width=\linewidth]{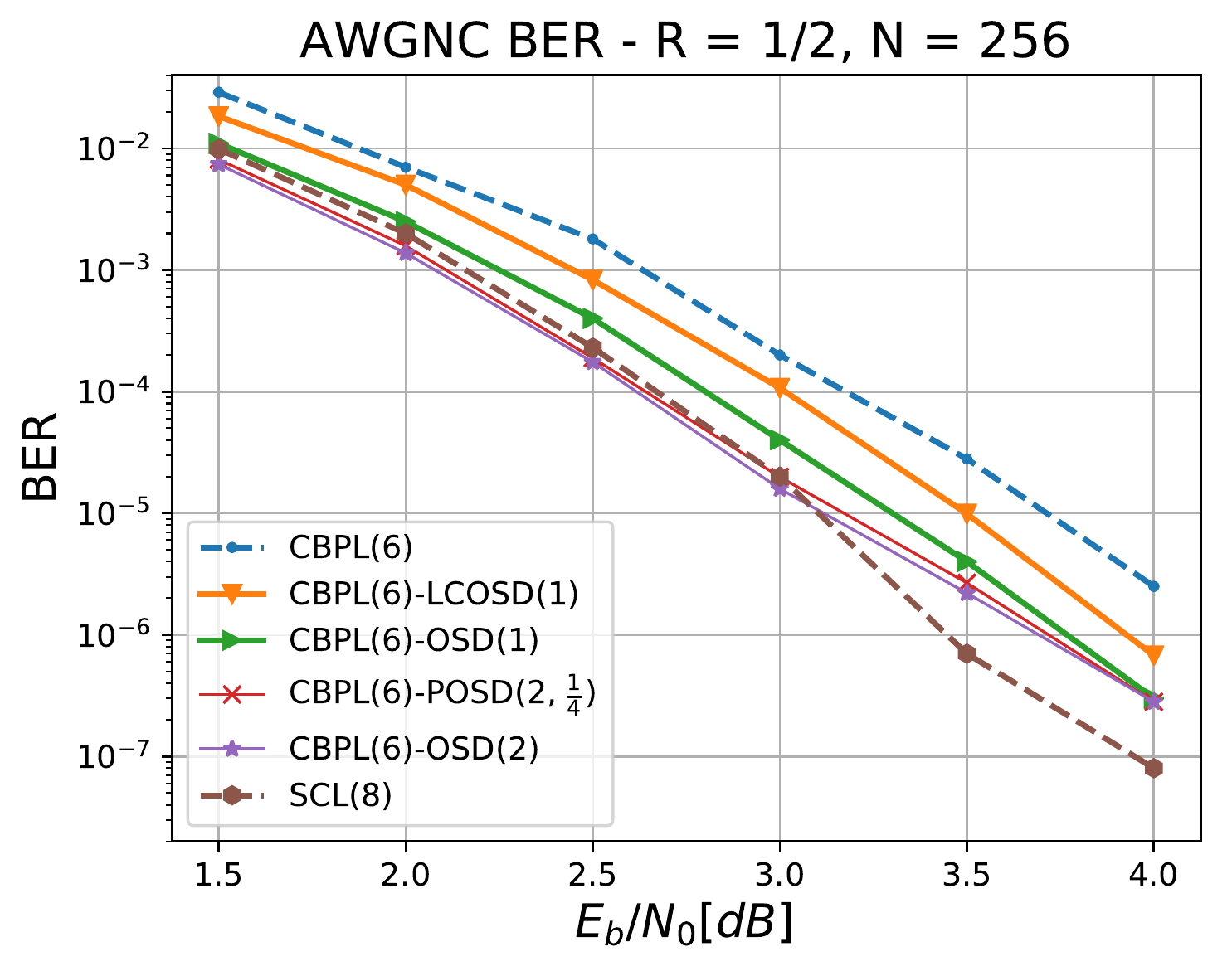}
	\end{subfigure}
	\hspace{0.006\textwidth}
	\begin{subfigure}[b]{.3\linewidth}
		\includegraphics[width=\linewidth]{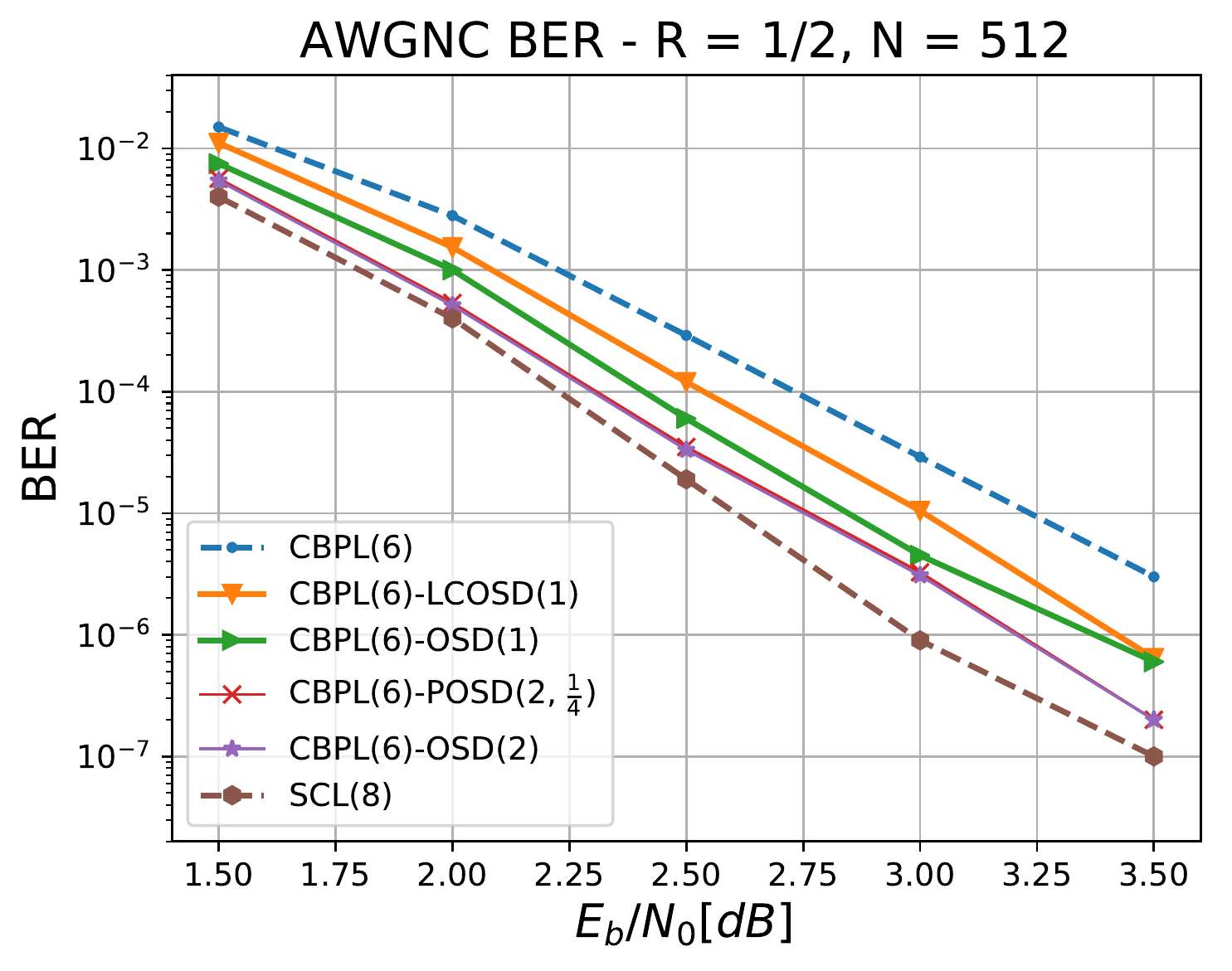}
	\end{subfigure}
	
	\vspace{2ex}
	
	\begin{subfigure}[b]{.3\linewidth}
		\includegraphics[width=\linewidth]{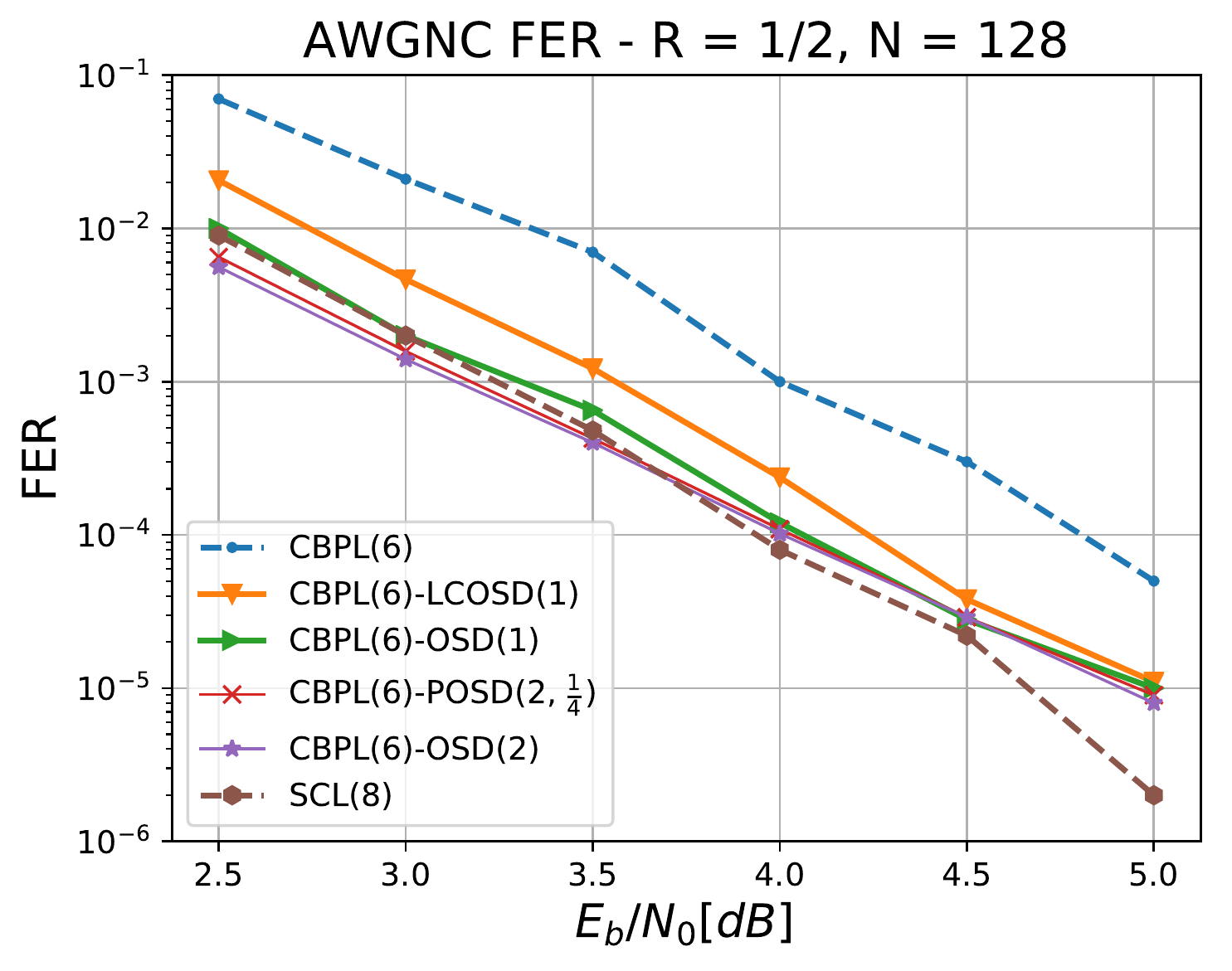}
	\end{subfigure}
	\hspace{0.006\textwidth}
	\begin{subfigure}[b]{.3\linewidth}
		\includegraphics[width=\linewidth]{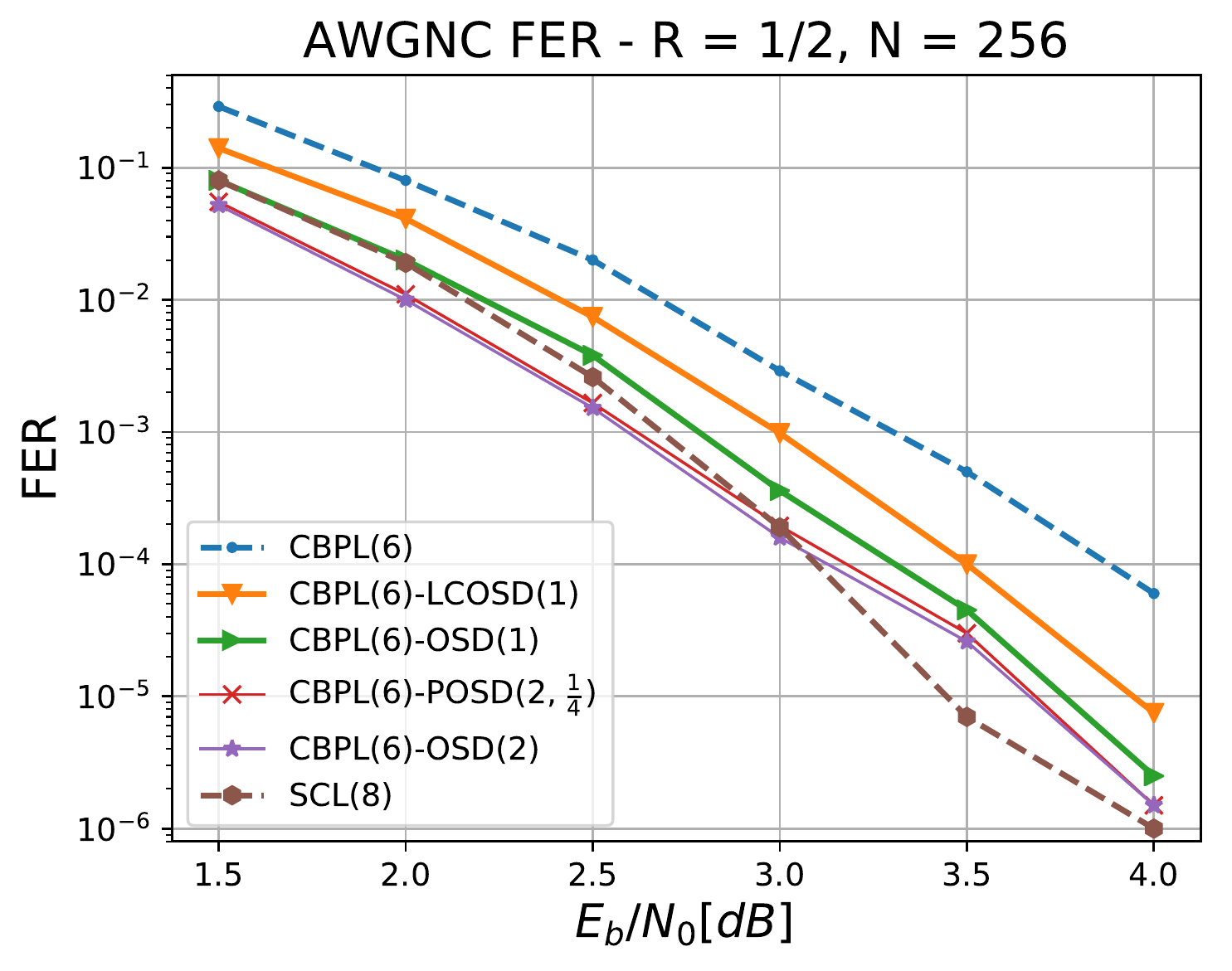}
	\end{subfigure}
	\hspace{0.006\textwidth}
	\begin{subfigure}[b]{.3\linewidth}
		\includegraphics[width=\linewidth]{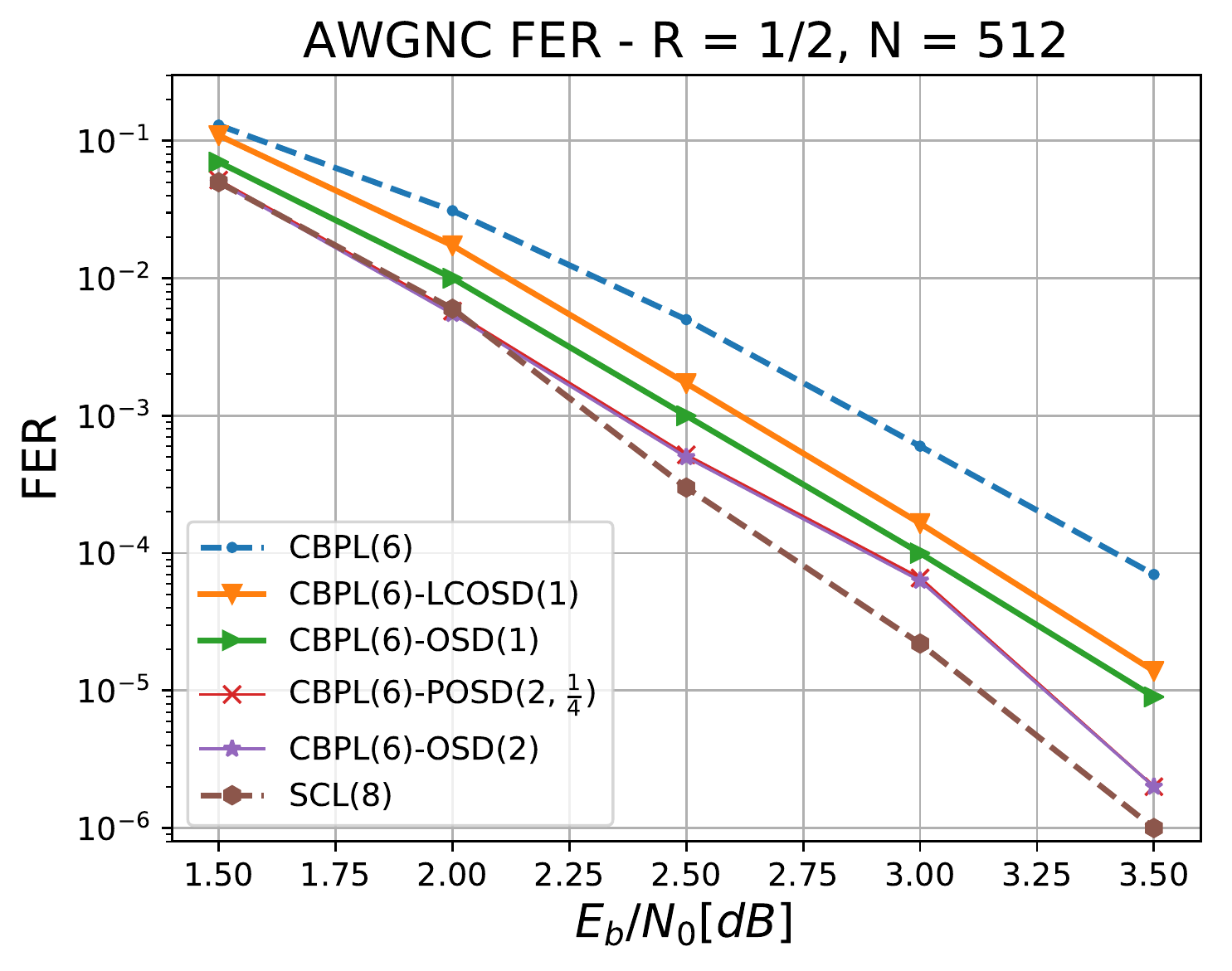}
	\end{subfigure}
	
	\vspace{2ex}
	
	\begin{subfigure}[b]{.3\linewidth}
		\includegraphics[width=\linewidth]{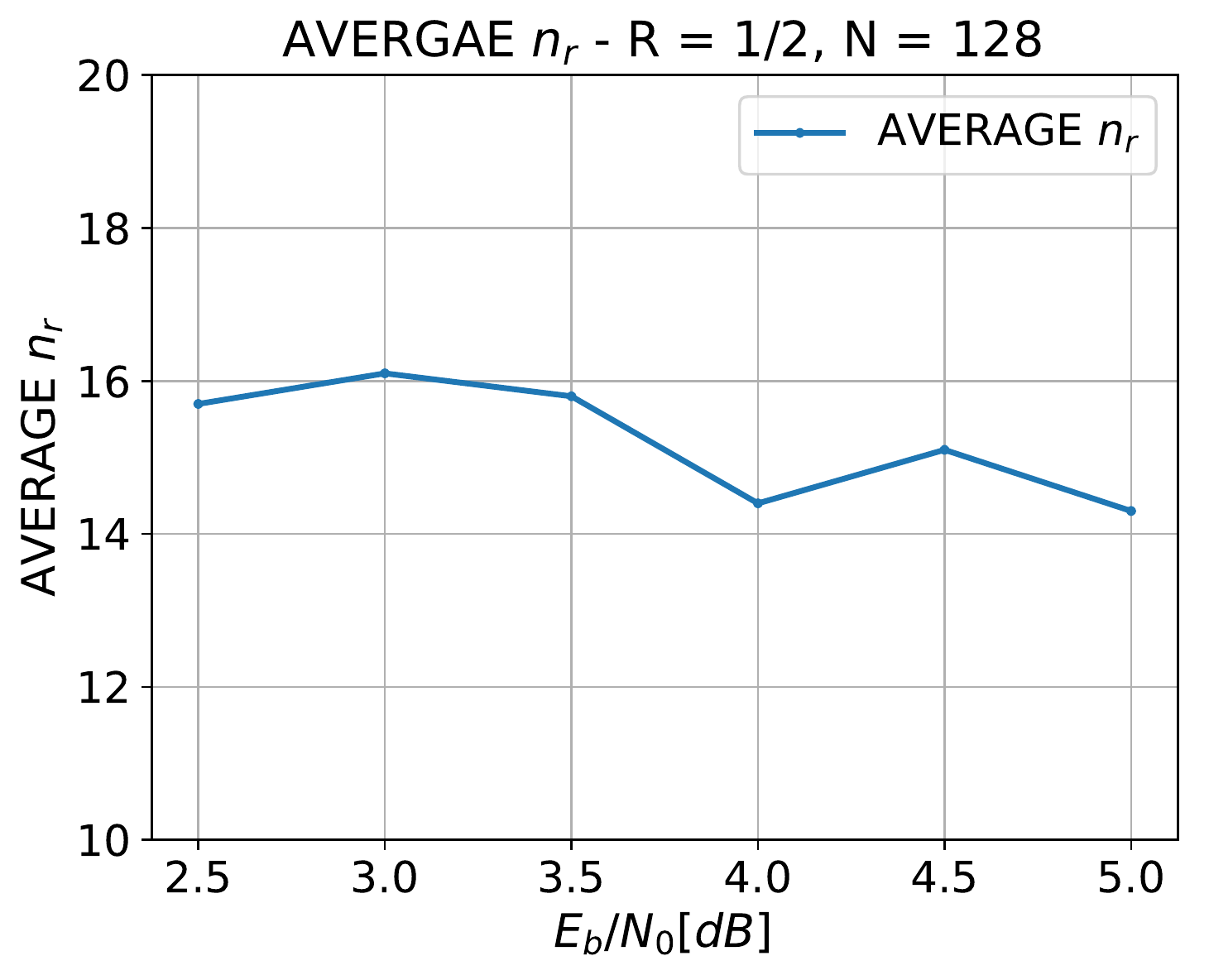}
	\end{subfigure}
	\hspace{0.006\textwidth}
	\begin{subfigure}[b]{.3\linewidth}
		\includegraphics[width=\linewidth]{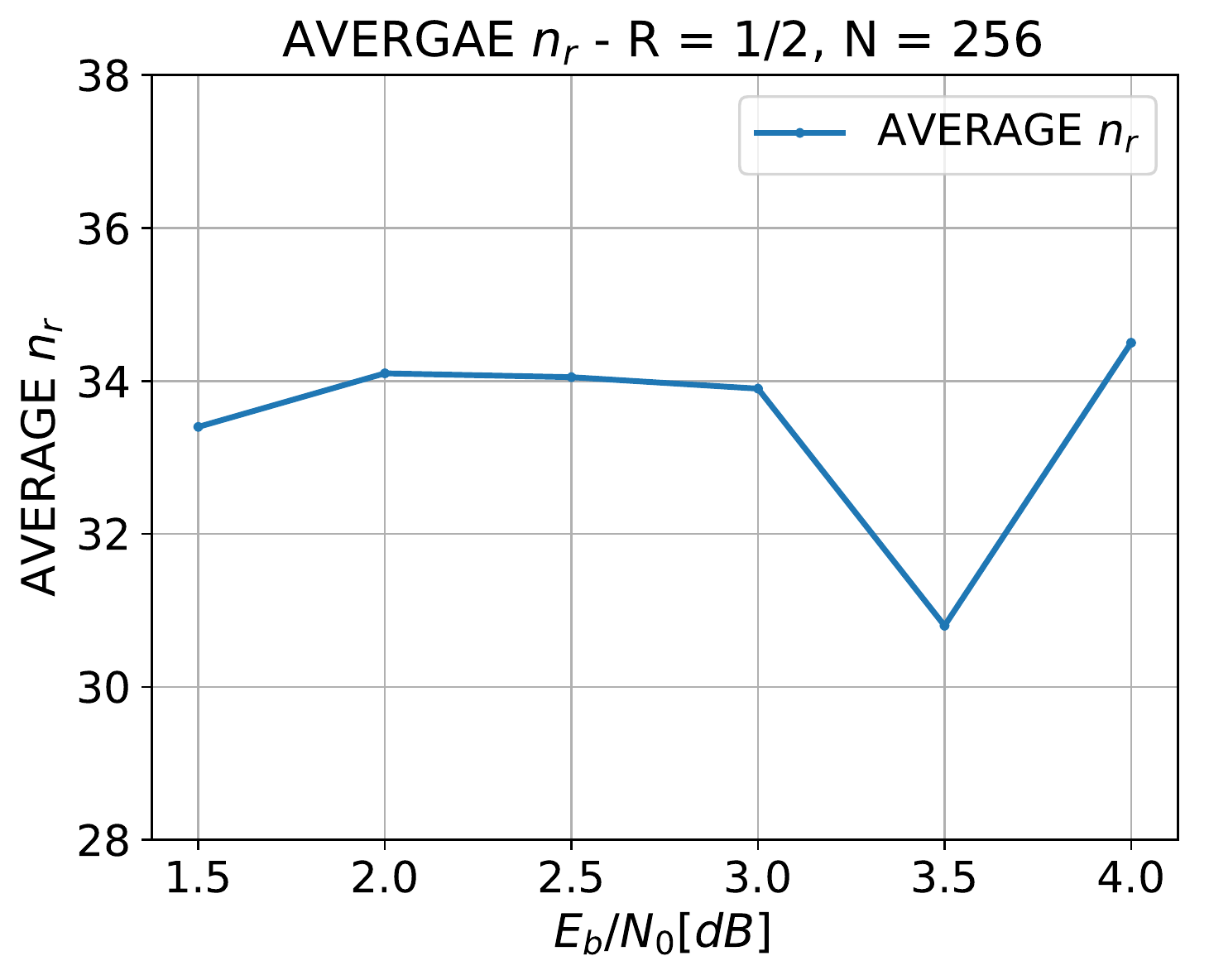}
	\end{subfigure}
	\hspace{0.006\textwidth}
	\begin{subfigure}[b]{.3\linewidth}
		\includegraphics[width=\linewidth]{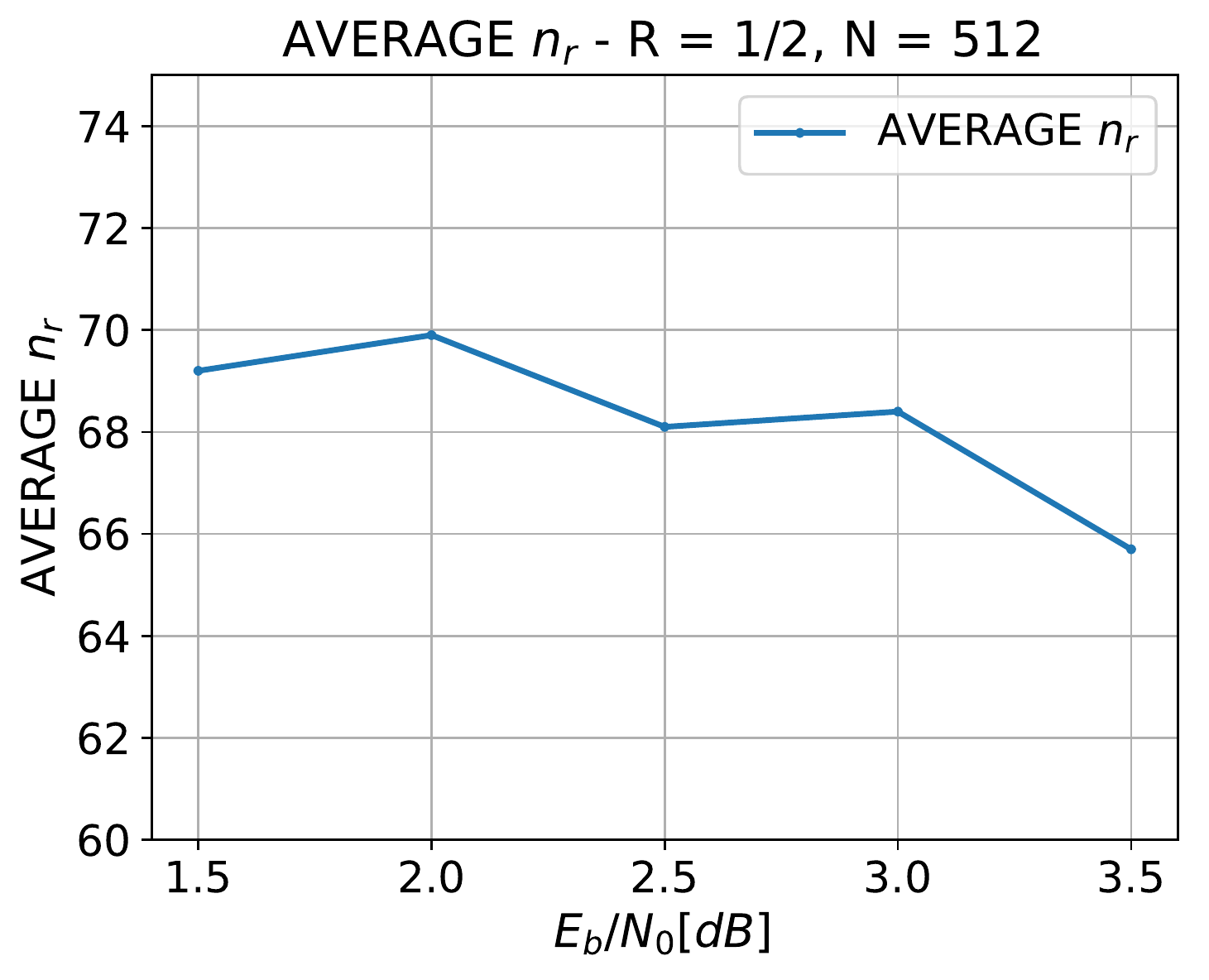}
	\end{subfigure}
	\caption{BER, FER, and average $n_r$, for CRC concatenated polar codes with rate $1/2$ and different blocklengths over the AWGNC.}
	\label{fig:OSDStats}
\end{figure}

\section{Conclusions} \label{sec:conclusions}
We have presented new decoders for (possibly CRC-augmented) polar codes transmitted over the BEC and AWGNC.
For the BEC, the algorithm computes the exact ML solution with average asymptotic computational complexity which is essentially the same as that of standard BP decoding.
For the AWGNC, we derived a novel decoder that combines CBPL and OSD decoding with a very low reprocessing order. A variant of the ML decoder derived for the BEC was used in order to efficiently implement the Gaussian elimination required by the first stage of OSD.
We demonstrated improved error rate performance and low computational complexity both for the BEC and for the AWGNC using numerical simulations.
For the AWGNC, CBPL-OSD(1) showed an improvement of $0.5-0.8$ dB compared to CBPL in the high SNR region.



\end{document}